\documentclass[12pt]{article}
\usepackage{a4}
\usepackage{amsfonts,amssymb,amsmath}
\usepackage{graphicx}
\usepackage{cite}

\newtheorem{theorem}{Theorem}

\newtheorem{lemma}[theorem]{Lemma}

\newenvironment{proof}[1][Proof]{\textbf{#1.} }{\ \rule{0.5em}{0.5em}}
\begin{document}

\title{Quantum mechanics with coordinate dependent noncommutativity}
\author{V.G. Kupriyanov\thanks{%
e-mail: vladislav.kupriyanov@gmail.com} \\
%EndAName
CMCC, Universidade Federal do ABC, Santo Andr\'{e}, SP, Brazil}
\date{\today                                        }
\maketitle

\begin{abstract}
Noncommutative quantum mechanics can be considered as a first step in the
construction of quantum field theory on noncommutative spaces of generic form, when the commutator between coordinates is a
function of these coordinates. In this paper we discuss the
mathematical framework of such a theory. The noncommutativity is treated as an external antisymmetric field satisfying the Jacoby identity. First, we propose a symplectic realization of a given Poisson manifold and construct the Darboux coordinates on the obtained symplectic manifold. Then we define the star product on a Poisson manifold and obtain the expression for the trace functional. The above ingredients are used to formulate a nonrelativistic quantum mechanics on noncommutative spaces of general form. All considered constructions are obtained as a formal series in the parameter of noncommutativity. In particular, the complete algebra of commutation relations between coordinates and conjugated momenta is a deformation of the standard Heisenberg algebra. As examples we consider a free particle and an isotropic harmonic oscillator on the rotational invariant noncommutative space.
\end{abstract}

\newpage
\section{Introduction}

Quantum field theory on noncommutative spaces has been studied extensively
during the last decades \cite{NCreviews}. The main attention was given to
the case of flat noncommutative space-time, realized by the coordinate
operators $\hat{x}^{i},\ i=1,...,N,$ satisfying the algebra $\left[ \hat{x}%
^{i},\hat {x}^{j}\right] =i\theta^{ij}\,,$ with a constant $\theta^{ij}$
being the parameters of noncommutativity. Different phenomenological
consequences of the presence of this type of noncommutativity in the theory
were studied. The comparison of the theoretical predictions with the
experimental data gives rise to the bounds of noncommutativity in this
model, see e.g. \cite{AGSV} and references therein. However, the restriction
to flat noncommutative spaces does not seem to be very natural. The physical
motivation for noncommutativity comes from a combination of the general
quantum mechanical arguments with Einstein relativity \cite{Doplicher}, and
space-time in any dynamical theory of gravity cannot be flat.

The presence of a more general type of noncommutativity, i.e., when the
parameters of noncommutativity depend on coordinates may lead to absolutely
different phenomenological consequences. For example, in \cite{BW} it was
shown that the Lagrangian of the Grosse-Wulkenhaar (renormalizable) model
\cite{GW} can be written as a Lagrangian of a scalar field propagating in a
curved noncommutative space, defined by the truncated Heisenberg algebra.
The problem is to construct a consistent quantum field theory on
noncommutative spaces of general form.

Our point of view is the following: one may construct noncommutative field theory
of generic form considering relativistic version of quantum mechanics (QM) with
coordinate operators satisfying the commutation relations
\begin{equation}
\left[ \hat{x}^{i},\hat{x}^{j}\right] =i\theta\hat{\omega}_{q}^{ij}\left(
\hat{x}\right) , \,\,\,i,j=1,...,N, \label{1}
\end{equation}
where $\hat{\omega}_{q}^{ij}\left( \hat{x}\right) $ is an operator defined
from physical considerations with a specified ordering which describes the
noncommutativity of the space. In this paper we discuss the mathematical framework of nonrelativistic
QM with coordinate operators satisfying commutation
relations of the type (\ref{1}). A two-dimensional model of position-dependent noncommutativity in QM was
proposed in \cite{GK}. The particular example was considered in \cite{Fring1}, where also was observed that canonical operators in these models are in general non Hermitian with respect to standard inner products \cite{Fling2}. For an example of QM and field theory on kappa-Minkowski space see e.g. \cite{Meljanac} and references therein. Quantum mechanical models on fuzzy spaces were discussed in \cite{Presnajder}.

To formulate a consistent QM on noncommutative spaces (\ref{1}) first we need to introduce momenta $\hat{p}_{i},$ conjugated to $\hat {x}^{i}$, i.e., to obtain the complete algebra of commutation relations, obeying the Jacobi identities, as a deformation in $\theta$ of the Heisenberg algebra%
\begin{align}
& \left[ \hat{x}^{i},\hat{x}^{j}\right] =i\theta\hat{\omega}_{q}^{ij}\left( \hat{x}\right)~,
\label{3} \\
&\left[ \hat{x}^{i},\hat{p}_{j}\right] =i\hat{\delta}_{q}^{ij}\left( \hat{x},\hat{p}\right)=\delta^{ij}+\theta\delta_{1}^{ij}\left( \hat{x},\hat{%
p}\right) +O\left( \theta^{2}\right) ~,  \notag \\
&\left[ \hat{p}_{i},\hat{p}_{j}\right] =i\hat{\varpi}_{q}^{ij}\left( \hat{x},\hat{p}\right)=\theta\varpi_{1}^{ij}\left( \hat{x},\hat{p}\right)
+O\left( \theta^{2}\right)~,  \notag
\end{align}
where $\delta^{ij}$ is the Kronecker delta. However, $\delta_{1}^{ij}\left(
\hat{x},\hat{p}\right) $ and $\varpi_{1}^{ij}\left( \hat{x},\hat{p}\right) $
are already an operator-valued function of $\hat{x}$ and $\hat{p}$. Then one should construct a representation of this algebra.
And finally, it is necessary to define a Hilbert space and to introduce an internal product on it $\langle\varphi|\psi\rangle$, so that the canonical operators $\hat{x}^{i}$ and $\hat{p}_{i}$ be self-adjoint, i.e., satisfy $\langle\hat{p}_{i}\varphi|\psi\rangle=\langle\varphi|\hat{p}_{i}\psi\rangle$ and $\langle\hat{x}^{i}\varphi|\psi\rangle=\langle\varphi|\hat{x}^{i}\psi\rangle$ for any two states $|\varphi\rangle$ and $|\psi\rangle$ from the Hilbert space.

In the Sec. 2 we show that the algebra (\ref{3}) can be obtained from a quantization of $2N$ dimensional symplectic manifold with coordinates $\xi
^{\mu }=\left( x^{i},p_{i}\right) ,\ \ \mu =1,...,2N,$ and a symplectic
structure
\begin{equation}
\Omega _{\mu \nu }\left( \xi \right) =\Omega _{\mu \nu }^{0}+O\left( \theta
\right) ,  \label{4a}
\end{equation}
such that $\Omega _{ij}=\theta \omega ^{ij}\left( x\right) $, where $\omega ^{ij}\left(
x\right) $ is a Poisson bi-vector, corresponding to the operator $\hat{\omega}_{q}^{ij}\left( \hat{x}\right)$ and $\Omega _{\mu \nu }^{0}$ is a canonical symplectic matrix. That is, to construct the complete algebra of commutation relations (\ref{3}) one should start with a finding of a symplectic realization of the corresponding Poisson manifold. In the Sec. 3 we give a recursive solution to this problem in a form
of power series in $\theta $.

Then, in the Sec. 4 we construct the Darboux coordinates $\eta ^{\mu
}=\left( y^{i},\pi _{i}\right) $ on the obtained symplectic manifold in a form of perturbative
series $\eta ^{\mu }\left( \xi \right) =\xi ^{\mu }+\theta \eta _{1}^{\mu
}\left( \xi \right) +O\left( \theta ^{2}\right) ,$ providing an explicit
formulas $\eta _{n}^{\mu }\left( \xi \right) $ of each order in $\theta $
and giving a complete description of the arbitrariness in our construction.
The particular case $y^{i}=x^{i}+\theta /2\omega ^{ij}\left( x\right)
p_{j}+O\left( \theta ^{2}\right) $ and $\pi _{i}=p_{i}$ is considered in the
Sec. 5, where we also present a direct method of finding of
Darboux coordinates starting from a Poisson bi-vector $\omega ^{ij}\left(
x\right) $.

In fact, it is shown that the problem of the construction of the symplectic
structure $\Omega $ in each order in $\theta $ is reduced to the solution of
algebraic equations. We give the explicit formulae for the
solution of these equations in each order in\ the deformation parameter. In
the Sec. 6 we discuss a general form of Darboux coordinates obtained from
the direct method:
\begin{equation*}
x^{i}=y^{i}-\frac{\theta }{2}\omega ^{il}\left( y\right) \pi _{l}+O\left(
\theta ^{2}\right) ,\,\,\,\ p_{i}=\pi _{i}-\theta j_{i}(y,\pi ,\theta ),
\end{equation*}%
where $j_{i}(y,\pi ,\theta )$ is an arbitrary vector.

Finally, we consider the canonical quantization of the obtained symplectic manifold. In Sec. 7 we use the expressions
for Darboux coordinates $\eta ^{\mu }\left( \xi \right) $ to
construct the polydifferential representation of the algebra (\ref{3}) and to define the star product on the corresponding Poisson manifold. The explicit form of the star product is given up to the third order in the deformation parameter $\theta$. To complete our construction we also need the expression for the trace functional on the algebra of star product. In Sec. 8 we describe the prerturbative procedure of the construction of the trace.

Using the above ingredients in Sec. 9 we formulate the quantum mechanics with noncommutative coordinates satisfying the algebra (\ref{1}), describing the Hilbert space, the internal product on it and the action of the canonical operators on the states from the Hilbert space.
As an example we solve the eigenvalue problem for the free particle on coordinate dependent NC space. Also we discuss the three-dimensional isotropic harmonic oscillator on rotational invariant NC space. This example shows that the noncommutativity can be introduced in a minimal way in the theory, i.e., one may obtain nonlocality without violating physical observables like the energy
spectrum, etc.

\section{Necessary and sufficient conditions}

Consider the noncommutative space defined by the commutation relations (\ref%
{1}). We choose the symmetric Weyl ordering for the operator $\hat {\omega}%
_{q}^{ij}\left( \hat{x}\right) $. Let $\omega_{q}^{ij}\left( x\right) $ be a
symbol of the operator $\hat{\omega}_{q}^{ij}\left( \hat {x}\right) $, i.e.,%
\begin{equation*}
\hat{\omega}_{q}^{ij}\left( \hat{x}\right) =\int\frac{d^{N}p}{\left(
2\pi\right) ^{N}}\tilde{\omega}_{q}^{ij}\left( p\right) e^{-ip_{m}\hat {x}%
^{m}},
\end{equation*}
where $\tilde{\omega}_{q}^{ij}(p)$ is a Fourier transform of ${\omega}_{q}^{ij}$.
The consistency condition for the algebra (\ref{1}) is a consequence of the
Jacobi identity (JI) and reads:%
\begin{equation}
\left[ \hat{x}^{i},\hat{\omega}_{q}^{jk}\right] +\left[ \hat{x}^{k},\hat{%
\omega}_{q}^{ij}\right] +\left[ \hat{x}^{j},\hat{\omega}_{q}^{ki}\right] =0.
\label{2}
\end{equation}
It implies, see e.g. \cite{KV}, that
\begin{equation}
\omega_{q}^{ij}\left( x\right) =\omega^{ij}\left( x\right) +\omega
_{co}^{ij}\left( x\right) ,  \label{omega}
\end{equation}
where $\omega^{ij}\left( x\right) $ should be a Poisson bi-vector, i.e.,
satisfy the equation%
\begin{equation}
\omega^{il}\partial_{l}\omega^{jk}+\omega^{kl}\partial_{l}\omega^{ij}+%
\omega^{jl}\partial_{l}\omega^{ki}=0,  \label{2a}
\end{equation}
and the term $\omega_{co}^{ij}\left( x\right) $ stands for non-Poisson
corrections to $\omega^{ij}\left( x\right) $ of higher order in $\theta$,
expressed in terms of $\omega^{ij}\left( x\right) $ and its derivatives,
which depend on specific ordering of the operator $\hat{\omega}%
_{q}^{ij}\left( \hat{x}\right) $. So, telling that the operator $\hat{\omega
}_{q}^{ij}\left( \hat{x}\right) $ is defined from physical considerations we
mean that the Poisson bi-vector $\omega^{ij}\left( x\right) $ is given and
the ordering is specified. In what follows we treat $\omega^{ij}\left(
x\right) $ as an external antisymmetric field obeying the eq. (\ref{2a}).

It is convenient to introduce the following notations:%
\begin{align*}
& \hat{\xi}^{\mu}=\left( \hat{x}^{i},\hat{p}_{i}\right) ,\ \ \mu =1,...,2N,
\\
& \hat{\Omega}_{\mu\nu}^{q}=\left(
\begin{array}{cc}
\theta\hat{\omega}_{q}^{ij} & \hat{\delta}_{q}^{ij} \\
-\hat{\delta}_{q}^{ji} & \hat{\varpi}_{q}^{ij}%
\end{array}
\right) .
\end{align*}
In these notations the algebra of commutation relations (\ref{3}) is written
as%
\begin{equation}
\left[ \hat{\xi}^{\mu},\hat{\xi}^{\nu}\right] =i\hat{\Omega}_{\mu\nu}^{q}~.
\label{3a}
\end{equation}
The Jacobi identity for (\ref{3a}) implies that%
\begin{equation}
\left[ \hat{\xi}^{\mu},\hat{\Omega}_{\nu\alpha}^{q}\right] +\left[ \hat {\xi}%
^{\alpha},\hat{\Omega}_{\mu\nu}^{q}\right] +\left[ \hat{\xi}^{\nu},\hat{%
\Omega}_{\alpha\mu}^{q}\right] =0~.  \label{cc}
\end{equation}
This equation we call the consistency condition for the algebra (\ref{3}).
Let $\Omega_{\mu\nu}^{q}\left( \xi\right) $ be a symbol of the operator $%
\hat{\Omega}_{\mu\nu}^{q}$. The eq. (\ref{cc}) means that
\begin{equation}
\Omega_{\mu\nu}^{q}\left( \xi\right) =\Omega_{\mu\nu}\left( \xi\right)
+\Omega_{\mu\nu}^{co}\left( \xi\right) ,  \label{Omega}
\end{equation}
where $\Omega_{\mu\nu}\left( \xi\right) $ should be a Poisson bi-vector,
i.e.,%
\begin{equation}
\Omega_{\mu\sigma}\partial_{\sigma}\Omega_{\nu\alpha}+\Omega_{\alpha\sigma
}\partial_{\sigma}\Omega_{\mu\nu}+\Omega_{\nu\sigma}\partial_{\sigma}%
\Omega_{\alpha\mu}=0,  \label{6}
\end{equation}
with $\partial_{\sigma}=\partial/\partial\xi^{\sigma}$, and $\Omega_{\mu\nu
}^{co}\left( \xi\right) $ is a non-Poisson correction due to ordering.
Comparing (\ref{Omega}) with (\ref{omega}) we conclude that
\begin{equation}
\Omega_{\mu\nu}\left( \xi\right) =\left(
\begin{array}{cc}
\theta\omega^{ij}\left( x\right) & \delta^{ij}\left( x,p\right) \\
-\delta^{ji}\left( x,p\right) & \varpi^{ij}\left( x,p\right)%
\end{array}
\right) ,  \label{7}
\end{equation}
where $\omega^{ij}\left( x\right) $ is a given Poisson bi-vector and the
functions $\delta^{ij}\left( x,p\right) $ and $\varpi^{ij}\left( x,p\right) $
to be determined from the equation (\ref{6}). Writing this equation in
components, e.g., $\mu=i,\ \nu=j,\ \alpha=N+k,$ one obtains partial
differential equations for the functions $\delta^{ij}\left( x,p\right) $ and
$\varpi^{ij}\left( x,p\right) $ in terms of given $\omega^{ij}\left(
x\right) $.

Thus, in order to obtain the expression for the operator $\hat{\Omega}_{\mu
\nu}^{q}$, satisfying the consistency condition (\ref{cc}), first we need to
find the Poisson bi-vector $\Omega_{\mu\nu}\left( \xi\right) ,$ such that $%
\Omega_{ij}\left( \xi\right) =\theta\omega^{ij}\left( x\right) ,$ which is a
necessary condition, and then determine the corrections due to ordering $%
\Omega_{\mu\nu}^{co}\left( \xi\right) $, which is a sufficient condition.

\section{Symplectic realization of a Poisson manifold}

This section is devoted to the solution of the necessary condition, i.e., to
the construction of the Poisson bi-vector $\Omega_{\mu\nu}\left( \xi\right) $%
. From a mathematical point of view the problem of finding a symplectic
realization of a Poisson manifold is the problem of the construction of a
local symplectic groupoid. The existence of a local symplectic groupoid for any
Poisson structure was shown in \cite{Karasev} and independently in \cite%
{Weinstein}. In the present paper we are interested in
explicit formulae for the symplectic structure $\Omega $ in a form (\ref{4a}%
), which is motivated by the condition that the complete algebra of commutation relations (\ref{3}) should be a deformation of the Heisenberg algebra.

Note that the condition (\ref{4a}) implies that the matrix $\Omega_{\mu\nu}$
can be represented as a perturbative series
\begin{equation}
\Omega_{\mu\nu}\left( \xi\right) =\sum_{n=0}^{\infty}\theta^{n}\Omega
_{\mu\nu}^{n}\left( \xi\right) ~,\ \ \ \Omega^{0}=\left(
\begin{array}{cc}
0 & 1_{N\times N} \\
-1_{N\times N} & 0%
\end{array}
\right) .  \label{9}
\end{equation}
In particular,
\begin{align}
\delta^{ij}\left( x,p\right) & =\delta^{ij}+\theta\delta_{1}^{ij}\left(
x,p\right) +O\left( \theta^{2}\right) ~,  \label{8} \\
\varpi^{ij}\left( x,p\right) & =\theta\varpi_{1}^{ij}\left( x,p\right)
+O\left( \theta^{2}\right) ~.  \notag
\end{align}
So, at least perturbativelly, $\Omega_{\mu\nu}$ is non-degenerate, $\det
\Omega\neq0.$ Let us denote its inverse matrix by $\bar{\Omega}_{\mu\nu}$.

From the identity $\bar{\Omega}_{\mu\sigma}\Omega_{\sigma\nu}=\delta_{\mu%
\nu} $ one obtains%
\begin{equation*}
\partial_{\sigma}\Omega_{\mu\nu}=-\Omega_{\mu\beta}\partial_{\sigma}\bar{%
\Omega}_{\beta\gamma}\Omega_{\gamma\nu}.
\end{equation*}
Taking this into account, the Jacobi identity (\ref{6}) can be rewritten as%
\begin{equation}
\Omega_{\alpha\sigma}\Omega_{\mu\beta}\Omega_{\nu\gamma}\left( \partial
_{\sigma}\bar{\Omega}_{\beta\gamma}+\partial_{\gamma}\bar{\Omega}_{\sigma
\beta}+\partial_{\beta}\bar{\Omega}_{\gamma\sigma}\right) =0.  \label{18}
\end{equation}
Which means that non-degenerate matrix $\Omega_{\mu\nu}$ obeys the JI (\ref%
{6}) if and only if its inverse obeys the equation%
\begin{equation}
\partial_{\sigma}\bar{\Omega}_{\beta\gamma}+\partial_{\gamma}\bar{\Omega }%
_{\sigma\beta}+\partial_{\beta}\bar{\Omega}_{\gamma\sigma}=0.  \label{19}
\end{equation}
If the matrix $\bar{\Omega}_{\mu\nu}$ has a form%
\begin{equation}
\bar{\Omega}_{\mu\nu}=\partial_{\mu}J_{\nu}-\partial_{\nu}J_{\mu},
\label{19a}
\end{equation}
where vector $J_{\mu}\left( \xi\right) $ is called symplectic potential,
then (\ref{19}) is automatically satisfied. We conclude that if the matrix $%
\Omega_{\mu\nu}$ obey the equation%
\begin{equation}
\partial_{\mu}J_{\nu}-\partial_{\nu}J_{\mu}=\Omega_{\mu\nu}^{-1},  \label{20}
\end{equation}
then it obeys the JI (\ref{6}).

Our idea is to study the equation (\ref{20}), where $\Omega_{ij}=\theta
\omega^{ij}\left( x\right) $ is given and satisfy the JI (\ref{2a}), while $%
\delta^{ij}$, $\varpi^{ij}$ and $J_{\mu}$ are unknown, instead of looking
for solution of eq. (\ref{6}). The eq. (\ref{6}) is an integrability
condition for (\ref{20}).

\subsection{Perturbative solution}

We are interested in a perturbative solution of (\ref{20}). To this end let
us first represent (\ref{20}) as a perturbative series and obtain equation
in each order in $\theta$. Since $\Omega_{\mu\nu}$ has a perturbative form (%
\ref{9}), its inverse matrix also can be written as
\begin{equation}
\bar{\Omega}_{\mu\nu}\left( \xi\right) =\sum_{n=0}^{\infty}\theta^{n}\bar{%
\Omega}_{\mu\nu}^{n}\left( \xi\right) ~,  \label{21}
\end{equation}
where $\bar{\Omega}^{0}=-\Omega^{0}$ and%
\begin{equation}
\bar{\Omega}_{\mu\nu}^{n}=\left(
\begin{array}{cc}
\bar{\omega}_{n}^{ij}\left( x,p\right) & \bar{\delta}_{n}^{ij}\left(
x,p\right) \\
-\bar{\delta}_{n}^{ji}\left( x,p\right) & \bar{\varpi}_{n}^{ij}\left(
x,p\right)%
\end{array}
\right) .  \label{21a}
\end{equation}
By the definition%
\begin{align}
& \bar{\omega}_{n}^{ij}=\partial_{i}J_{j}^{n}-\partial_{j}J_{i}^{n},
\label{21b} \\
& \bar{\delta}_{n}^{ij}=\partial_{i}J_{N+j}^{n}-\partial_{j}^{p}J_{i}^{n},
\notag \\
& \bar{\varpi}_{n}^{ij}=\partial_{i}^{p}J_{N+j}^{n}-%
\partial_{j}^{p}J_{N+i}^{n},  \notag
\end{align}
where $\partial_{i}=\partial/\partial x^{i}$ and $\partial_{i}^{p}=\partial/%
\partial p^{i}$.

Substituting (\ref{9}) and (\ref{21}) in the identity $\bar{\Omega}_{\mu
\sigma}\Omega_{\sigma\nu}=\delta_{\mu\nu}$ one has%
\begin{equation}
\left( \sum_{n=0}^{\infty}\theta^{n}\bar{\Omega}^{n}\right) \left(
\sum_{m=0}^{\infty}\theta^{m}\Omega^{m}\right) =\sum_{n=0}^{\infty}\theta
^{n}\sum_{m=0}^{n}\bar{\Omega}^{n-m}\Omega^{m}=1.  \label{22}
\end{equation}
Equating the powers in $\theta$ in the right and in the left hand sides of (%
\ref{22}) we obtain%
\begin{equation}
-\bar{\Omega}^{n}\Omega^{0}=\Omega^{0}\Omega^{n}+\sum_{m=1}^{n-1}\bar{\Omega }%
^{n-m}\Omega^{m},\ \ \ n\geq1,
\end{equation}
which can be rewritten as
\begin{equation}
\Omega^{0}\bar{\Omega}^{n}\Omega^{0}=\Omega^{n}+\sum_{m=1}^{n-1}\Omega^{0}\bar{\Omega }%
^{n-m}\Omega^{0}\Omega^{0}\Omega^{m},\ \ \ n\geq1.  \label{23}
\end{equation}
Then, iterating (\ref{23}) we come to the following equation%
\begin{align}
&-\Omega^{0} \bar{\Omega}^{n}\Omega^{0}+\Omega^{n}+\sum_{m_{1}=1}^{n-1}%
\Omega^{n-m_{1}}\Omega^{0}\Omega^{m_{1}}+...+  \label{24} \\
&
\sum_{m_{1}=1}^{n-r}...\sum_{m_{r}=1}^{n-m_{1}-...-m_{r-1}-1}%
\Omega^{n-m_{1}-...-m_{r}}\Omega^{0}\Omega^{m_{r}}\Omega^{0}...\Omega
^{0}\Omega^{m_{1}}+...+\Omega^{1}\left( \Omega^{0}\Omega^{1}\right) ^{n-1}=
\notag \\
& -\Omega^{0}\bar{\Omega}^{n}\Omega^{0}+\Omega^{n}+\sum_{r=1}^{n-1}\left\{
\sum_{m_{1}=1}^{n-r}...\sum_{m_{r}=1}^{n-m_{1}-...-m_{r-1}-1}\Omega
^{n-m_{1}-...-m_{r}}\Omega^{0}\Omega^{m_{r}}\Omega^{0}...\Omega^{0}%
\Omega^{m_{1}}\right\}=0 .  \notag
\end{align}
In particular, for $n=1$ (\ref{24}) means%
\begin{equation}
-\Omega^{0}\bar{\Omega}^{1}\Omega^{0}+\Omega^{1}=0,  \label{25}
\end{equation}
for $n=2$:%
\begin{equation}
-\Omega^{0}\bar{\Omega}^{2}\Omega^{0}+\Omega^{2}+\Omega^{1}\Omega^{0}\Omega
^{1}=0,  \label{26}
\end{equation}
etc.

Eq. (\ref{25}) in components reads%
\begin{align}
& \partial_{i}^{p}J_{N+j}^{1}-\partial_{j}^{p}J_{N+i}^{1}=-\omega^{ij}\left(
x\right) ,  \label{27} \\
& \partial_{i}J_{N+j}^{1}-\partial_{j}^{p}J_{i}^{1}=\delta_{1}^{ji}\left(
x,p\right) ,  \label{28} \\
& \partial_{i}J_{j}^{1}-\partial_{j}J_{i}^{1}=-\varpi_{1}^{ij}\left(
x,p\right) ,  \label{29}
\end{align}
where $\omega^{ij}\left( x\right) $ is given and $\delta_{1}^{ji}$ and $%
\varpi_{1}^{ij}$ should be found. From the eq. (\ref{27}) one defines, $%
J_{N+i}^{1}=\frac{1}{2}\omega^{ij}\left( x\right)
p_{j}+\partial_{i}^{p}f^{1}\left( x,p\right) $ in terms of $%
\omega^{ij}\left( x\right) $ and arbitrary function $f^{1}\left( x,p\right) $%
. $J_{i}^{1}$ remains arbitrary and eqs. (\ref{28}) and (\ref{29}) define $%
\delta_{1}^{ji}$ and $\varpi _{1}^{ij}$ in terms of $\omega^{ij},\ f^{1}$
and $J_{i}^{1}$.

Eq. (\ref{26}) in components reads%
\begin{align}
&
\partial_{i}^{p}J_{N+j}^{2}-\partial_{j}^{p}J_{N+i}^{2}-\omega^{il}%
\delta_{1}^{jl}-\delta_{1}^{il}\omega^{lj}=0,  \label{30} \\
&
\delta_{2}^{ij}=\partial_{i}^{p}J_{j}^{2}-\partial_{j}J_{N+i}^{2}+%
\delta_{1}^{il}\delta_{1}^{lj}-\omega^{il}\varpi_{1}^{lj},  \label{31} \\
&
\varpi_{2}^{ij}=\partial_{j}J_{i}^{2}-\partial_{i}J_{j}^{2}+\delta_{1}^{li}%
\varpi_{1}^{lj}+\varpi_{1}^{il}\delta_{1}^{lj}.  \label{32}
\end{align}
Again, eq. (\ref{30}) is a differential equation on $J_{N+i}^{2}$ which
defines it up to the gradient of an arbitrary function $f^{2}\left(
x,p\right) $, while eqs. (\ref{31}) and (\ref{32}) are the definition of the
functions $\delta_{2}^{ij}$ and $\varpi_{2}^{ij}$ in terms of $J_{N+i}^{2}$,
arbitrary $J_{i}^{2}$ and first order functions.

The integrability condition for the eq. (\ref{30}) is%
\begin{equation}
\partial_{k}^{p}\left( \omega^{il}\delta_{1}^{jl}+\delta_{1}^{il}\omega
^{lj}\right) +\mbox{cycl.}(ijk)=0.  \label{33}
\end{equation}
We rewrite it as%
\begin{equation}
\omega^{jl}\left( \partial_{i}^{p}\delta_{1}^{kl}-\partial_{k}^{p}\delta
_{1}^{il}\right) +\mbox{cycl.}(ijk)=0.  \label{34}
\end{equation}
Then, using (\ref{27}) and (\ref{28}) one finds%
\begin{equation}
\partial_{i}^{p}\delta_{1}^{kl}-\partial_{k}^{p}\delta_{1}^{il}=\partial
_{l}\omega^{ki}.  \label{35}
\end{equation}
That is, (\ref{33}) is exactly the JI (\ref{2a}) for the matrix $\omega
^{ij}\left( x\right) .$

Now let us write the eq. (\ref{24}) in components in the $n$-th order in $%
\theta$, for $n>1$:%
\begin{align}
& \partial_{i}^{p}J_{N+j}^{n}-\partial_{j}^{p}J_{N+i}^{n}=  \label{36} \\
& -\sum_{r=1}^{n-1}\left\{
\sum_{m_{1}=1}^{n-r}...\sum_{m_{r}=1}^{n-m_{1}-...-m_{r-1}-1}\left[
\Omega^{n-m_{1}-...-m_{r}}\Omega^{0}\Omega^{m_{r}}\Omega^{0}...\Omega^{0}%
\Omega^{m_{1}}\right] _{ij}\right\} ,  \notag \\
& \delta_{n}^{ij}=-\bar{\delta}_{n}^{ji}  \label{37} \\
& -\sum_{r=1}^{n-1}\left\{
\sum_{m_{1}=1}^{n-r}...\sum_{m_{r}=1}^{n-m_{1}-...-m_{r-1}-1}\left[
\Omega^{n-m_{1}-...-m_{r}}\Omega^{0}\Omega^{m_{r}}\Omega^{0}...\Omega^{0}%
\Omega^{m_{1}}\right] _{iN+j}\right\} ,  \notag \\
& \varpi_{n}^{ij}=-\overline{\varpi}_{n}^{ij}  \label{38} \\
& -\sum_{r=1}^{n-1}\left\{
\sum_{m_{1}=1}^{n-r}...\sum_{m_{r}=1}^{n-m_{1}-...-m_{r-1}-1}\left[
\Omega^{n-m_{1}-...-m_{r}}\Omega^{0}\Omega^{m_{r}}\Omega^{0}...\Omega^{0}%
\Omega^{m_{1}}\right] _{N+iN+j}\right\} ,  \notag
\end{align}
where functions $\bar{\delta}_{n}^{ji}$ and $\bar{\varpi}_{n}^{ij}$ are
determined in (\ref{21b}). The logic is the same as in the first two orders,
from the eq. (\ref{36}) one finds $J_{N+i}^{n}$ up to the gradient of an
arbitrary function $f^{n}\left( x,p\right) $, and then uses the eqs. (\ref%
{37}) and (\ref{38}) to define functions $\delta_{n}^{ij}$ and $%
\varpi_{n}^{ij}$ in terms of $J_{N+i}^{n}$, arbitrary $J_{i}^{n}$ and
functions $\delta_{m}^{ij}$ and $\varpi_{m}^{ij}$ of lower orders, $%
m<n$.

\subsection{Integrability condition in the $n$-th order}

Let us study the integrability condition for the eq. (\ref{36}). Since $%
\partial_{i}^{p}=\partial_{N+i}=\Omega_{i\sigma}^{0}\partial_{\sigma}$, it
can be written as%
\begin{align}
C= & \Omega_{k\sigma}^{0}\partial_{\sigma}\sum_{r=1}^{n-1}\left\{
\sum_{m_{1}=1}^{n-r}...\sum_{m_{r}=1}^{n-m_{1}-...-m_{r-1}-1}\left[
\Omega^{n-m_{1}-...-m_{r}}\Omega^{0}\Omega^{m_{r}}\Omega^{0}...\Omega
^{0}\Omega^{m_{1}}\right] _{ij}\right\}  \label{39} \\
& +\mbox{cycl.}(ijk)=0.  \notag
\end{align}
In the first order the integrability condition was satisfied automatically,
since right hand side of the eq. (\ref{27}) does not depend on $p$. In the
second order the integrability condition was satisfied as a consequence of
the JI (\ref{2a}) for the matrix $\omega^{ij}\left( x\right) $. To prove the
existence of the symplectic potential $J_{\mu}\left( \xi\right) $ we should
prove that the integrability condition (\ref{39}) is satisfied in all orders.

We will do it by the induction. Suppose that the eq. (\ref{36}) is solvable
up to the $\left( n-1\right) $-th order. It means that both the symplectic
potential $J_{\mu}$ and symplectic structure $\Omega_{\mu\nu}$ can be
defined from the eq. (\ref{24}) up to the $\left( n-1\right) $-th order.
That is, JI (\ref{6}) is satisfied up to the $\left( n-1\right) $-th order,%
\begin{align}
& \Omega_{\mu\sigma}^{0}\partial_{\sigma}\Omega_{\nu\alpha}^{r}+\sum
_{s=1}^{r-1}\Omega_{\mu\sigma}^{r-s}\partial_{\sigma}\Omega_{\nu\alpha}^{s}+%
\mbox{cycl.}(\mu\nu\alpha)=0,\ \ r=2,...,n-1,  \label{40} \\
& \Omega_{\mu\sigma}^{0}\partial_{\sigma}\Omega_{\nu\alpha}^{1}+\mbox{cycl.}%
(\mu\nu\alpha)=0.  \notag
\end{align}

Let us show that (\ref{39}) is satisfied. To begin with we prove the
following Lemma.

\begin{lemma}
\label{L1} If $J_{\mu}$ exists up to the $\left( n-3\right) $-th order, then
\begin{align}
& D=  \label{a2} \\
& \sum_{r=3}^{n-1}\left\{
\sum_{m_{1}=1}^{n-r}...\sum_{m_{r}=1}^{n-...-m_{r-1}-1}\left[
\Omega^{m_{r-1}}\Omega^{0}\right] _{j\alpha}\Omega_{i\sigma}^{m_{r}}%
\partial_{\sigma}\left[ \Omega^{n-...-m_{r}}\Omega^{0}\Omega^{m_{r-2}}...%
\Omega^{0}\Omega^{m_{2}}\right] _{\alpha\beta }\left[ \Omega^{0}%
\Omega^{m_{1}}\right] _{\beta k}\right\}  \notag \\
& +\mbox{cycl.}(ijk)=  \notag\\
&\sum_{m_{1}=1}^{n-3}\sum_{m_{2}=1}^{n-m_{1}-2}%
\sum_{m_{3}=1}^{n-m_{1}-m_{2}}\left[ \Omega^{m_{2}}\Omega^{0}\right]
_{j\alpha}\Omega_{i\sigma}^{m_3}\partial_{\sigma}\Omega_{\alpha%
\beta}^{n-...-m_{3}}\left[ \Omega^{0}\Omega^{m_{1}}\right] _{\beta k}\notag\\
& +\sum_{m_{1}=1}^{n-4}...\sum_{m_{4}=1}^{n-...-m_{3}-1}\left[ \Omega
^{m_{3}}\Omega^{0}\right] _{j\alpha}\Omega_{i\sigma}^{m_{4}}\partial_{\sigma
}\left[ \Omega^{n-...-m_{4}}\Omega^{0}\Omega^{m_{2}}\right] _{\alpha\beta }%
\left[ \Omega^{0}\Omega^{m_{1}}\right] _{\beta k}  \notag \\
& +...+\sum_{m_{1}=1}^{n-r}...\sum_{m_{r}=1}^{n-...-m_{r-1}-1}\left[
\Omega^{m_{r-1}}\Omega^{0}\right] _{j\alpha}\Omega_{i\sigma}^{m_{r}}%
\partial_{\sigma}\left[ \Omega^{n-...-m_{r}}\Omega^{0}\Omega^{m_{r-2}}...%
\Omega^{0}\Omega^{m_{2}}\right] _{\alpha\beta}\left[ \Omega^{0}\Omega^{m_{1}}%
\right] _{\beta k}  \notag \\
& +...+\left[ \Omega^{1}\Omega^{0}\right] _{j\alpha}\Omega_{i\sigma}^{1}%
\partial_{\sigma}\left[ \Omega^{1}\left( \Omega^{0}\Omega^{1}\right) ^{n-4}%
\right] _{\alpha\beta}\left[ \Omega^{0}\Omega^{1}\right] _{\beta k}+%
\mbox{cycl.}(ijk)=0.  \notag
\end{align}
\end{lemma}

\begin{proof}
Using the eq. (\ref{24})  the first term of (\ref{a2}) can be written as
\begin{align}
&\sum_{m_{1}=1}^{n-3}\sum_{m_{2}=1}^{n-m_{1}-2}%
\sum_{m_{3}=1}^{n-m_{1}-m_{2}}\left[ \Omega^{m_{2}}\Omega^{0}\right]
_{j\alpha}\Omega_{i\sigma}^{m_3}\partial_{\sigma}\Omega_{\alpha%
\beta}^{n-...-m_{3}}\left[ \Omega^{0}\Omega^{m_{1}}\right] _{\beta k}+%
\mbox{cycl.}(ijk)=\label{a31}\\
&
\sum_{m_{1}=1}^{n-3}\sum_{m_{2}=1}^{n-m_{1}-1}%
\sum_{m_{3}=1}^{n-m_{1}-m_{2}-1}\Omega_{j\alpha}^{m_{1}}\Omega_{\beta
k}^{m_{2}}\Omega_{i\sigma}^{m_{3}}\partial_{\sigma}\bar{\Omega}%
_{\alpha\beta}^{n-m_{1}-m_{2}-m_{3}}\notag  \\
& -\sum_{m_{1}=1}^{n-4}...\sum_{m_{4}=1}^{n-...-m_{3}-1}\left[ \Omega
^{m_{3}}\Omega^{0}\right] _{j\alpha}\Omega_{i\sigma}^{m_{4}}\partial_{\sigma
}\left[ \Omega^{n-...-m_{4}}\Omega^{0}\Omega^{m_{2}}\right] _{\alpha\beta }%
\left[ \Omega^{0}\Omega^{m_{1}}\right] _{\beta k}  \notag \\
& -...-\sum_{m_{1}=1}^{n-r}...\sum_{m_{r}=1}^{n-...-m_{r-1}-1}\left[
\Omega^{m_{r-1}}\Omega^{0}\right] _{j\alpha}\Omega_{i\sigma}^{m_{r}}%
\partial_{\sigma}\left[ \Omega^{n-...-m_{r}}\Omega^{0}\Omega^{m_{r-2}}...%
\Omega^{0}\Omega^{m_{2}}\right] _{\alpha\beta}\left[ \Omega^{0}\Omega^{m_{1}}%
\right] _{\beta k}  \notag \\
& -...-\left[ \Omega^{1}\Omega^{0}\right] _{j\alpha}\Omega_{i\sigma}^{1}%
\partial_{\sigma}\left[ \Omega^{1}\left( \Omega^{0}\Omega^{1}\right) ^{n-4}%
\right] _{\alpha\beta}\left[ \Omega^{0}\Omega^{1}\right] _{\beta k}+%
\mbox{cycl.}(ijk).  \notag
\end{align}
Now, if to substitute (\ref{a31}) back to (\ref{a2}) one finds that
\begin{align*}
&D=\sum_{m_{1}=1}^{n-3}\sum_{m_{2}=1}^{n-m_{1}-1}%
\sum_{m_{3}=1}^{n-m_{1}-m_{2}-1}\Omega_{j\alpha}^{m_{1}}\Omega_{\beta
k}^{m_{2}}\Omega_{i\sigma}^{m_{3}}\partial_{\sigma}\bar{\Omega}%
_{\alpha\beta}^{n-m_{1}-m_{2}-m_{3}}+\mbox{cycl.}(ijk)=  \\
&-\sum_{m_{1}=1}^{n-3}\sum_{m_{2}=1}^{n-m_{1}-1}%
\sum_{m_{3}=1}^{n-m_{1}-m_{2}-1}\Omega_{i\sigma}^{m_{1}}\Omega_{j%
\alpha}^{m_{2}}\Omega_{k\beta }^{m_{3}}\left( \partial_{\sigma}\bar{\Omega}%
_{\alpha\beta}^{n-m_{1}-m_{2}-m_{3}}+\mbox{cycl.}(\sigma\alpha\beta)\right)
 =0,
\end{align*}
which complete the proof.
\end{proof}

The sum $C+D$ can be represented as
\begin{equation}
C+D=\sum_{r=1}^{n-1}T_{r},  \label{a4}
\end{equation}
where%
\begin{align}
& T_{1}=\sum_{m_{1}=1}^{n-1}\Omega_{k\sigma}^{0}\partial_{\sigma}\left[
\Omega^{n-m_{1}}\Omega^{0}\Omega^{m_{1}}\right] _{ij}+\mbox{cycl.}(ijk),
\label{a5} \\
&
T_{2}=\sum_{m_{1}=1}^{n-2}\sum_{m_{2}=1}^{n-m_{1}-1}\Omega_{i\sigma}^{0}%
\partial_{\sigma}\left[ \Omega^{n-m_{1}-m_{2}}\Omega^{0}\Omega^{m_{2}}%
\Omega^{0}\Omega^{m_{1}}\right] _{ij}+\mbox{cycl.}(ijk),  \label{a6} \\
& T_{r}=\sum_{m_{1}=1}^{n-r}...\sum_{m_{r}=1}^{n-m_{1}-...-m_{r-1}-1}\left\{
\Omega_{k\sigma}^{0}\partial_{\sigma}\left[ \Omega^{n-m_{1}-...-m_{r}}%
\Omega^{0}\Omega^{m_{r}}\Omega^{0}...\Omega^{0}\Omega^{m_{1}}\right]
_{ij}\right.  \label{a7} \\
& \left. +\left[ \Omega^{m_{r-1}}\Omega^{0}\right] _{j\alpha}\Omega_{i%
\sigma}^{m_{r}}\partial_{\sigma}\left[ \Omega^{n-...-m_{r}}\Omega^{0}%
\Omega^{m_{r-2}}...\Omega^{0}\Omega^{m_{2}}\right] _{\alpha\beta }\left[
\Omega^{0}\Omega^{m_{1}}\right] _{\beta k}\right\}  \notag \\
& \left. +\mbox{cycl.}(ijk),\ \ \ 3\leq r\leq n-1.\right.  \notag
\end{align}

Let us introduce the following notations:%
\begin{align}
& I\left( n,1\right)
=\sum_{m_{1}=1}^{n-1}\Omega_{k\sigma}^{m_{1}}\partial_{\sigma}%
\Omega_{ij}^{n-m_{1}}+\mbox{cycl.}(ijk),  \label{a8} \\
& I\left( n,r\right) =  \label{a9} \\
& \sum_{m_{1}=1}^{n-r}...\sum_{m_{r}=1}^{n-...-m_{r-1}-1}\left\{ \left(
\Omega_{i\sigma}^{m_{r}}\partial_{\sigma}\Omega_{j\alpha}^{n-...-m_{r}}+%
\mbox{cycl.}(i\alpha\beta)\right) \left[ \Omega^{0}\Omega^{m_{r-1}}...%
\Omega^{0}\Omega^{m_{1}}\right] _{\alpha k}+...+\right.  \notag \\
& \left[ \Omega^{m_{r-1}}\Omega^{0}...\Omega^{0}\Omega^{m_{r-s}}\Omega ^{0}%
\right] _{j\alpha}\left( \Omega_{\beta\sigma}^{m_{r-1}}\partial_{\sigma
}\Omega_{i\alpha}^{n-...-m_{r-1}}+\mbox{cycl.}(i\alpha\beta)\right) \left[
\Omega^{0}\Omega^{m_{r-s-1}}\Omega^{0}...\Omega^{0}\Omega^{m_{1}}\right]
_{\beta k}  \notag \\
& \left. +...+\left[ \Omega^{m_{r-1}}\Omega^{0}...\Omega^{0}\Omega^{m_{2}}%
\Omega^{0}\right] _{j\alpha}\left(
\Omega_{\beta\sigma}^{m_{r-1}}\partial_{\sigma}\Omega_{i%
\alpha}^{n-...-m_{r-1}}+\mbox{cycl.}(i\alpha \beta)\right) \left[
\Omega^{0}\Omega^{m_{1}}\right] _{\beta k}\right\}  \notag \\
& \left. +\mbox{cycl.}(ijk),\ \ \ 2\leq r\leq n-1,\right.  \notag \\
& I\left( n,n\right) =\left(
\Omega_{i\sigma}^{0}\partial_{\sigma}\Omega_{j\alpha}^{1}+\mbox{cycl.}%
(ij\alpha)\right) \left[ \left( \Omega ^{0}\Omega^{1}\right) ^{n-1}\right]
_{\alpha k}  \label{a10} \\
& +\sum_{s=1}^{n-3}\left[ \left( \Omega^{1}\Omega^{0}\right) ^{s}\right]
_{j\alpha}\left( \Omega_{i\sigma}^{0}\partial_{\sigma}\Omega_{\alpha\beta
}^{1}+\mbox{cycl.}(i\alpha\beta)\right) \left[ \left( \Omega^{0}\Omega
^{1}\right) ^{n-s-1}\right] _{\beta k}+\mbox{cycl.}(ijk).  \notag
\end{align}
In particular,%
\begin{equation*}
I\left( n,2\right) =\sum_{m_{1}=1}^{n-2}\sum_{m_{2}=1}^{n-m_{1}-1}\left(
\Omega_{i\sigma}^{m_{2}}\partial_{\sigma}\Omega_{j\alpha}^{n-m_{1}-m_{2}}+%
\mbox{cycl.}(ij\alpha)\right) \left[ \Omega^{0}\Omega^{m_{1}}\right]
_{\alpha k}+\mbox{cycl.}(ijk),
\end{equation*}
etc.

Note that since $\Omega_{ij}^{m}=0$ for $m>1$ and $\Omega_{ij}^{1}=\omega
^{ij}\left( x\right) $ does not depend on $p$, one has%
\begin{equation}
I\left( n,1\right)
=\sum_{m_{1}=1}^{n-1}\Omega_{k\sigma}^{m_{1}}\partial_{\sigma}%
\Omega_{ij}^{n-m_{1}}+\mbox{cycl.}(ijk)=\omega^{kl}\partial_{l}\omega^{ij}+%
\mbox{cycl.}(ijk)=0.  \label{a11}
\end{equation}
Also, due to Jacobi identity (\ref{40})%
\begin{equation}
I\left( n,n\right) =0.  \label{a12}
\end{equation}

\begin{lemma}
\label{L2} If symplectic potential $J_{\mu}$ exists up to the $n-1$ order,
then
\begin{equation}
T_{r}=I\left( n,r\right) -I\left( n,r+1\right) ,\ \ 1\leq r\leq n-1.
\label{a13}
\end{equation}
\end{lemma}

\begin{proof}
First we write $T_{1}$ as%
\begin{equation}
T_{1}=\sum_{m_{1}=1}^{n-1}\left(
\Omega_{i\sigma}^{0}\partial_{\sigma}\Omega_{j\alpha}^{n-m_{1}}+\Omega_{j%
\sigma}^{0}\partial_{\sigma}\Omega_{\alpha i}^{n-m_{1}}\right) \left[
\Omega^{0}\Omega^{m_{1}}\right] _{\alpha k}+\mbox{cycl.}(ijk).  \label{a14}
\end{equation}
Then using (\ref{40}) we represent it in the form%
\begin{align*}
&
T_{1}=-\sum_{m_{1}=1}^{n-1}\Omega_{\alpha\sigma}^{0}\partial_{\sigma}%
\Omega_{ij}^{n-m_{1}}\left[ \Omega^{0}\Omega^{m_{1}}\right] _{\alpha k} \\
& -\sum_{m_{1}=1}^{n-2}\sum_{m_{2}=1}^{n-m_{1}-1}\left( \Omega_{i\sigma
}^{m_{2}}\partial_{\sigma}\Omega_{j\alpha}^{n-m_{1}-m_{2}}+\mbox{cycl.}%
(ij\alpha)\right) \left[ \Omega^{0}\Omega^{m_{1}}\right] _{\alpha k}+%
\mbox{cycl.}(ijk).
\end{align*}
The sum in the second line is exactly $I\left( n,2\right) $. Since,
\begin{equation}
\Omega_{\alpha\sigma}^{0}\Omega_{\alpha\beta}^{0}=\delta_{\sigma\beta},
\label{delta}
\end{equation}
the sum in the first line can be rewritten as%
\begin{equation*}
\sum_{m_{1}=1}^{n-1}\Omega_{k\sigma}^{m_{1}}\partial_{\sigma}\Omega
_{ij}^{n-m_{1}}+\mbox{cycl.}(ijk).
\end{equation*}
That is,
\begin{equation}
T_{1}=I\left( n,1\right) -I\left( n,2\right) .  \label{a15}
\end{equation}
For the second term we write%
\begin{align}
& T_{2}=\sum_{m_{1}=1}^{n-2}\sum_{m_{2}=1}^{n-m_{1}-1}\left\{ \left(
\Omega_{i\sigma}^{0}\partial_{\sigma}\Omega_{j\alpha}^{n-m_{1}-m_{2}}+%
\Omega_{j\sigma}^{0}\partial_{\sigma}\Omega_{\alpha
i}^{n-m_{1}-m_{2}}\right) \left[ \Omega^{0}\Omega^{m_{2}}\Omega^{0}%
\Omega^{m_{1}}\right] _{\alpha k}\right.  \label{a16} \\
& \left. +\left[ \Omega^{m_{2}}\Omega^{0}\right] _{j\alpha}\Omega
_{i\sigma}^{0}\partial_{\sigma}\Omega_{\alpha\beta}^{n-m_{1}-...-m_{k}}\left[
\Omega^{0}\Omega^{m_{1}}\right] _{\beta k}\right\} +\mbox{cycl.}(ijk)  \notag
\end{align}
Applying the eq. (\ref{40}) in (\ref{a16}) one transforms it as%
\begin{align}
& \sum_{m_{1}=1}^{n-2}\sum_{m_{2}=1}^{n-m_{1}-1}\left\{ -\Omega
_{\alpha\sigma}^{0}\partial_{\sigma}\Omega_{ij}^{n-m_{1}-m_{2}}\left[
\Omega^{0}\Omega^{m_{2}}\Omega^{0}\Omega^{m_{1}}\right] _{\alpha k}\right.
\notag \\
& \left. -\left[ \Omega^{m_{2}}\Omega^{0}\right] _{j\alpha}\left(
\Omega_{\beta\sigma}^{0}\partial_{\sigma}\Omega_{i\alpha}^{n-m_{1}-m_{2}}+%
\Omega_{\alpha\sigma}^{0}\partial_{\sigma}\Omega_{\beta
i}^{n-m_{1}-m_{2}}\right) \left[ \Omega^{0}\Omega^{m_{1}}\right] _{\beta
k}\right\}  \notag \\
&
-\sum_{m_{1}=1}^{n-3}\sum_{m_{2}=1}^{n-m_{1}-1}%
\sum_{m_{3}=1}^{n-m_{1}-m_{2}-1}\left\{ \left(
\Omega_{i\sigma}^{3}\partial_{\sigma}\Omega_{j\alpha}^{n-m_{1}-m_{2}-1}+%
\mbox{cycl.}(i\alpha\beta)\right) \left[ \Omega^{0}\Omega^{m_{2}}\Omega^{0}%
\Omega^{m_{1}}\right] _{\alpha k}\right.  \notag \\
& \left. +\left[ \Omega^{m_{2}}\Omega^{0}\right] _{j\alpha}\left(
\Omega_{\beta\sigma}^{m_{3}}\partial_{\sigma}\Omega_{i%
\alpha}^{n-m_{1}-m_{2}-1}+\mbox{cycl.}(i\alpha\beta)\right) \left[
\Omega^{0}\Omega^{m_{1}}\right] _{\beta k}\right\} +\mbox{cycl.}(ijk).
\notag
\end{align}
The last sum here is $I\left( n,3\right) .$ Using the identity (\ref{delta})
we simplify the first sum and see that%
\begin{equation}
T_{2}=I\left( n,2\right) -I\left( n,3\right) .  \label{a17}
\end{equation}
For the terms $T_{r}$ with $3\leq r\leq n-1$ the logic is the same as in the
first two cases, but we should also take into account the contribution from
the sum (\ref{a2}). First using Jacobi identity (\ref{40}) in each order we
represent it as%
\begin{align*}
& T_{r}=\sum_{m_{1}=1}^{n-r}...\sum_{m_{r}=1}^{n-m_{1}-...-m_{r-1}-1}\left\{
-\Omega_{\alpha\sigma}^{0}\partial_{\sigma}\Omega_{ij}^{n-m_{1}-...-m_{r}}
\left[ \Omega^{0}\Omega^{m_{r}}\Omega^{0}...\Omega^{0}\Omega^{m_{1}}\right]
_{\alpha k}\right. \\
& -...-\left[ \Omega^{m_{r}}\Omega^{0}...\Omega^{0}\Omega^{m_{r-s}}\Omega^{0}%
\right] _{j\alpha}\left( \Omega_{\beta\sigma}^{0}\partial_{\sigma
}\Omega_{i\alpha}^{n-m_{1}-...-m_{r}}\right. \\
& \left. +\Omega_{\alpha\sigma}^{0}\partial_{\sigma}\Omega_{\beta
i}^{n-m_{1}-...-m_{r}}\right) \left[ \Omega^{0}\Omega^{m_{r-s-1}}\Omega
^{0}...\Omega^{0}\Omega^{m_{1}}\right] _{\beta k} \\
& +\left[ \Omega^{m_{r-1}}\Omega^{0}...\Omega^{0}\Omega^{m_{r-s}}\Omega ^{0}%
\right] _{j\alpha}\Omega_{i\sigma}^{m_{r}}\partial_{\sigma}\Omega
_{\alpha\beta}^{n-m_{1}-...-m_{r}}\left[ \Omega^{0}\Omega^{m_{r-s-1}}%
\Omega^{0}...\Omega^{0}\Omega^{m_{1}}\right] _{\beta k} \\
& -...-\left[ \Omega^{m_{r}}\Omega^{0}...\Omega^{0}\Omega^{m_{2}}\Omega ^{0}%
\right] _{j\alpha}\left(
\Omega_{\beta\sigma}^{0}\partial_{\sigma}\Omega_{i%
\alpha}^{n-m_{1}-...-m_{r}}\right. \\
& \left. +\Omega_{\alpha\sigma}^{0}\partial_{\sigma}\Omega_{\beta
i}^{n-m_{1}-...-m_{r}}\right) \left[ \Omega^{0}\Omega^{m_{1}}\right] _{\beta
k} \\
& \left. +\left[ \Omega^{m_{r-1}}\Omega^{0}...\Omega^{0}\Omega^{m_{2}}%
\Omega^{0}\right] _{j\alpha}\Omega_{i\sigma}^{m_{r}}\partial_{\sigma}%
\Omega_{\alpha\beta}^{n-m_{1}-...-m_{r}}\left[ \Omega^{0}\Omega^{m_{1}}%
\right] _{\beta k}\right\} +\mbox{cycl.}(ijk) \\
& -I\left( n,r+1\right) .
\end{align*}
Then simplifying the first sum we end up with%
\begin{equation*}
T_{r}=I\left( n,r\right) -I\left( n,r+1\right) ,\ \ 3\leq r\leq n-1.
\end{equation*}
\end{proof}

Using this Lemma and eqs. (\ref{a11}) and (\ref{a12}) we represent the sum (%
\ref{a4}) as%
\begin{equation}
C+D=\sum_{r=1}^{n-1}\left[ I\left( n,r\right) -I\left( n,r+1\right) \right]
=I\left( n,1\right) -I\left( n,n\right) =0.
\end{equation}
Since $D=0$ due to Lemma \ref{L1}, we conclude that $C=0$. That is, by the
induction, the integrability condition (\ref{39}) holds true in all orders.

\subsection{Concluding remarks}

Once (\ref{39}) is satisfied in all orders, from the eqs. (\ref{36})-(\ref%
{38}) we can define $J_{N+i}^{n}$, $\delta_{n}^{ij}$ and $\varpi _{n}^{ij}$
for any $n$ and construct the perturbative solution for the eq. (\ref{20}).
The arbitrariness in the construction of functions $\delta ^{ij}\left(
x,p\right) $ and $\varpi^{ij}\left( x,p\right) $ is described by an
arbitrary vector $J_{i}\left( x,p,\theta\right) $ and arbitrary function $%
f\left( x,p,\theta\right) $.

Note that since we have a symplectic structure $\Omega _{\mu \nu }$ obtained
from a symplectic potential $J_{\mu }$, one can always write a first-order
action leading to this symplectic structure%
\begin{equation}
S=\int dt\left[ J_{\mu }\left( \xi \right) \dot{\xi}^{\mu }-H\left( \xi
\right) \right] ,  \label{S}
\end{equation}%
where $H\left( \xi \right) $ is a Hamiltonian of a system. If an external
field $\omega ^{ij}\left( x\right) $ transforms as a tensor and $J_{i}\left(
x,p,\theta \right) $ transforms as a vector with respect to the Lorentz
group, one can also write a relativistic generalization of the action (\ref%
{S}), like it was done, e.g., in \cite{GKS} for the case of spin
noncommutativity and in \cite{GitKup} for canonical noncommutativity.

\section{Darboux coordinates}

Since $\Omega_{\mu\nu}\left( \xi\right) $ is a Poisson bi-vector, we may
define the Poisson bracket (PB) as%
\begin{equation}
\left\{ \xi^{\mu},\xi^{\nu}\right\} =\Omega_{\mu\nu}\left( \xi\right) .
\label{11}
\end{equation}
The construction of a quantum algebra (\ref{3}) is, in fact, the
quantization of a classical system with a Poisson Bracket (\ref{11}).

One of the methods of quantization of classical systems with noncanonical PB
is a quantization in Darboux coordinates. That is, first one should change
the phase space coordinates: $\xi_{\mu}=\left( x^{i},p_{i}\right)
\rightarrow \eta_{\mu}=\left( y^{i},\pi_{i}\right) $: $\xi_{\mu}=\xi_{\mu}%
\left( \eta\right),$ where new variables have canonical PB:%
\begin{equation}
\left\{ \eta_{\mu},\eta_{\nu}\right\} =\Omega_{\mu\nu}^{0},  \label{41}
\end{equation}
and are called Darboux coordinates. And then construct quantization in these
new phase space coordinates. Since
the existence of a Darboux coordinates is guarantied by the Darboux theorem
\cite{Arnold}, here we are interested in an explicit expression for them.

We will search for a Darboux
coordinates in a perturbative form:%
\begin{equation}
\xi_{\mu}\left( \eta\right)
=\eta_{\mu}+\sum_{n=1}^{\infty}\theta^{n}\xi_{\mu}^{n}\left( \eta\right) .
\label{42}
\end{equation}
The inverse transformation reads%
\begin{equation}
\eta_{\mu}\left( \xi\right)
=\xi_{\mu}+\sum_{n=1}^{\infty}\theta^{n}\eta_{\mu}^{n}\left( \xi\right) .
\label{43}
\end{equation}
Note that once (\ref{43}) is known, the inverse transformation (\ref{42})
may be easily calculated from the algebraic equation $\eta_{\mu}\left(
\xi\left( \eta\right) \right) =\eta_{\mu}$. In particular,
\begin{equation*}
\xi_{\mu}\left( \eta\right) =\eta_{\mu}-\theta\eta_{\mu}^{1}\left(
\eta\right) +\theta^{2}\left( \eta_{\nu}^{1}\left( \eta\right)
\partial_{\nu}\eta_{\mu}^{1}\left( \eta\right) -\eta_{\mu}^{2}\left(
\eta\right) \right) +O\left( \theta^{3}\right) .
\end{equation*}
So, we will look for (\ref{43}) and then construct (\ref{42}). To obtain
differential equation on $\eta_{\mu}\left( \xi\right) $ let us write (\ref%
{41}) as%
\begin{equation*}
\left\{ \eta_{\mu}\left( \xi\right) ,\eta_{\nu}\left( \xi\right) \right\}
=\Omega_{\mu\nu}^{0},
\end{equation*}
and take into account (\ref{11}), so%
\begin{equation}
\partial_{\sigma}\eta_{\mu}\Omega_{\sigma\rho}\partial_{\rho}\eta_{\nu}=%
\Omega_{\mu\nu}^{0}.  \label{44}
\end{equation}
This equation is just a consequence of the definition of a Darboux
coordinates. Substituting in (\ref{44}) perturbative expressions (\ref{9})
and (\ref{43}) one has%
\begin{equation}
\left(
\delta_{\sigma\mu}+\sum_{n=1}^{\infty}\theta^{n}\partial_{\sigma}\eta_{%
\mu}^{n}\right) \left( \Omega_{\sigma\rho}^{0}+\sum_{n=1}^{\infty
}\theta^{n}\Omega_{\sigma\rho}^{n}\right) \left( \delta_{\rho\nu}+\sum
_{n=1}^{\infty}\theta^{n}\partial_{\rho}\eta_{\nu}^{n}\right) =\Omega_{\mu
\nu}^{0}.  \label{45}
\end{equation}
From this equation, equating the powers in $\theta$ in the right and in the
left hand sides, we obtain the differential equations on the functions $%
\eta_{\mu}^{n}$ in each order $n$. Thus, in the first order one has%
\begin{equation}
\partial_{\sigma}\eta_{\mu}^{1}\Omega_{\sigma\nu}^{0}+\Omega_{\mu\sigma}^{0}%
\partial_{\sigma}\eta_{\nu}^{1}+\Omega_{\mu\nu}^{1}=0.  \label{46}
\end{equation}
In the second order we have%
\begin{equation}
\partial_{\sigma}\eta_{\mu}^{2}\Omega_{\sigma\nu}^{0}+\Omega_{\mu\sigma}^{0}%
\partial_{\sigma}\eta_{\nu}^{2}+\Omega_{\mu\nu}^{2}+\partial_{\sigma}\eta_{%
\mu}^{1}\Omega_{\sigma\rho}^{0}\partial_{\rho}\eta_{\nu}^{1}+\partial_{%
\sigma}\eta_{\mu}^{1}\Omega_{\sigma\nu}^{1}+\Omega_{\mu\sigma}^{1}\partial_{%
\sigma}\eta_{\nu}^{1}=0.  \label{47}
\end{equation}
The equation in the $n$-th order reads%
\begin{equation}
\partial_{\sigma}\eta_{\mu}^{n}\Omega_{\sigma\nu}^{0}+\Omega_{\mu\sigma}^{0}%
\partial_{\sigma}\eta_{\nu}^{n}+F_{\mu\nu}^{n}=0,  \label{49}
\end{equation}
where%
\begin{align}
& F_{\mu\nu}^{1}=\Omega_{\mu\nu}^{1},  \label{48} \\
&
F_{\mu\nu}^{2}=\Omega_{\mu\nu}^{2}+\partial_{\sigma}\eta_{\mu}^{1}\Omega_{%
\sigma\rho}^{0}\partial_{\rho}\eta_{\nu}^{1}+\partial_{\sigma}\eta_{\mu}^{1}%
\Omega_{\sigma\nu}^{1}+\Omega_{\mu\sigma}^{1}\partial_{\sigma
}\eta_{\nu}^{1},  \notag \\
& F_{\mu\nu}^{n}=\Omega_{\mu\nu}^{n}+\sum_{m=1}^{n-1}\left[ \partial
_{\sigma}\eta_{\mu}^{n-m}\Omega_{\sigma\rho}^{0}\partial_{\rho}\eta_{%
\nu}^{m}+\partial_{\sigma}\eta_{\mu}^{n-m}\Omega_{\sigma\nu}^{m}+\Omega_{\mu
\sigma}^{n-m}\partial_{\sigma}\eta_{\nu}^{m}\right]  \notag \\
&
+\sum_{m=1}^{n-2}\sum_{k=1}^{n-m-1}\Omega_{\sigma\rho}^{n-m-k}\partial_{%
\sigma}\eta_{\mu}^{k}\partial_{\rho}\eta_{\nu}^{m}=0,\ \ \ n\geq 3.  \notag
\end{align}

\begin{lemma}
\label{L3} The integrability condition for the equation (\ref{49}) is
\begin{equation}
\Omega_{\alpha\beta}^{0}\partial_{\beta}F_{\mu\nu}^{n}+\mathrm{cycl.}%
(\alpha\mu\nu)=0.  \label{50}
\end{equation}
\end{lemma}

\begin{proof}
One can easily verify that
\begin{equation}
\Omega_{\alpha\beta}^{0}\partial_{\beta}\left( \partial_{\sigma}\eta_{\mu
}^{n}\Omega_{\sigma\nu}^{0}+\Omega_{\mu\sigma}^{0}\partial_{\sigma}\eta_{\nu
}^{n}\right) +\mbox{cycl.}(\alpha\mu\nu)=0~,
\end{equation}
so (\ref{50}) is a necessary condition for a solution of eq. (\ref{49}). To
prove, that (\ref{50}) is a sufficient condition let us write (\ref{49}) in
components. Note that $\Omega_{i\sigma}^{0}\partial_{\sigma}=%
\partial_{i}^{p} $ and $\Omega_{N+i\sigma}^{0}\partial_{\sigma}=-%
\partial_{i} $. For $\mu=i$ and $\nu=j$, (\ref{49}) has a form%
\begin{equation}
\partial_{i}^{p}y_{n}^{j}-\partial_{j}^{p}y_{n}^{i}=F_{ij}^{n}.  \label{51}
\end{equation}
Its integrability condition is
\begin{equation}
\partial_{i}^{p}F_{jk}^{n}+\mbox{cycl.}(ijk)=0,  \label{52}
\end{equation}
which is exactly (\ref{50}) for $\alpha=i,\ \mu=j$ and $\nu=k$. For $\mu=N+i$
and $\nu=N+j$, (\ref{49}) reads%
\begin{equation}
\partial_{j}\pi_{i}^{n}-\partial_{i}\pi_{j}^{n}=F_{N+iN+j}^{n},  \label{51a}
\end{equation}
with the integrability condition
\begin{equation}
\partial_{i}F_{N+jN+k}^{n}+\mbox{cycl.}(ijk)=0,  \label{52a}
\end{equation}
which is eq. (\ref{50}) for $\alpha=N+i,\ \mu=N+j$ and $\nu=N+k$. Finally,
for $\mu=N+i$ and $\nu=j$, eq. (\ref{49}) is%
\begin{equation}
-\partial_{i}y_{n}^{j}-\partial_{j}^{p}\pi_{i}^{n}=F_{N+ij}^{n}.  \label{53}
\end{equation}
If treat this equation as an equation for $y_{n}^{j}$:%
\begin{equation*}
\partial_{i}y_{n}^{j}=-\partial_{j}^{p}\pi_{i}^{n}-F_{N+ij}^{n},
\end{equation*}
then its integrability condition is
\begin{align}
\partial_{k}\left( \partial_{j}^{p}\pi_{i}^{n}+F_{N+ij}^{n}\right)
-\partial_{i}\left( \partial_{j}^{p}\pi_{k}^{n}+F_{N+kj}^{n}\right) & =
\label{54} \\
-\partial_{i}F_{N+kj}^{n}+\partial_{j}^{p}F_{N+iN+k}^{n}-%
\partial_{k}F_{jN+i}^{n} & =0.  \notag
\end{align}
If treat (\ref{53}) as an equation for $\pi_{i}^{n}$:%
\begin{equation*}
\partial_{j}^{p}\pi_{i}^{n}=-\partial_{i}y_{n}^{j}-F_{N+ij}^{n},
\end{equation*}
its integrability condition is
\begin{align}
\partial_{k}^{p}\left( \partial_{i}y_{n}^{j}+F_{N+ij}^{n}\right)
-\partial_{j}^{p}\left( \partial_{i}y_{n}^{k}+F_{N+ik}^{n}\right) & =
\label{55} \\
\partial_{i}F_{kj}^{n}+\partial_{j}^{p}F_{kN+i}^{n}+%
\partial_{k}^{p}F_{N+ij}^{n} & =0.  \notag
\end{align}
Equations (\ref{54}) and (\ref{55}) are (\ref{50}) for $\alpha=N+i,\ \mu=N+k$%
, $\nu=j,$ and for $\alpha=N+i,\ \mu=k$, $\nu=j$ correspondingly. So, the
eq. (\ref{49}) written in components gives eqs. (\ref{51}), (\ref{51a}) and (%
\ref{53}) with integrability conditions (\ref{52}), (\ref{52a}), (\ref{54})
and (\ref{55}), which are exactly eq. (\ref{50}) written in components.
\end{proof}

For the eq. (\ref{46}) in the first order the condition (\ref{50}) is
\begin{equation}
\Omega_{\alpha\beta}^{0}\partial_{\beta}\Omega_{\mu\nu}^{1}+\mathrm{cycl.}%
(\alpha\mu\nu)=0.  \label{56}
\end{equation}
This equation is a Jacobi identity (\ref{6}) in the first order. In the
second order in $\theta$ (\ref{50}) reads%
\begin{equation}
\Omega_{\alpha\beta}^{0}\partial_{\beta}\left(
\Omega_{\mu\nu}^{2}+\partial_{\sigma}\eta_{\mu}^{1}\Omega_{\sigma\rho}^{0}%
\partial_{\rho}\eta_{\nu}^{1}+\partial_{\sigma}\eta_{\mu}^{1}\Omega_{\sigma%
\nu}^{1}+\Omega_{\mu\sigma}^{1}\partial_{\sigma}\eta_{\nu}^{1}\right) +%
\mathrm{cycl.}(\alpha\mu\nu)=0.  \label{57}
\end{equation}
First we calculate%
\begin{align}
& \Omega_{\alpha\beta}^{0}\partial_{\beta}\left( \partial_{\sigma}\eta_{\mu
}^{1}\Omega_{\sigma\rho}^{0}\partial_{\rho}\eta_{\nu}^{1}+\partial_{\sigma
}\eta_{\mu}^{1}\Omega_{\sigma\nu}^{1}+\Omega_{\mu\sigma}^{1}\partial_{\sigma
}\eta_{\nu}^{1}\right) +\mathrm{cycl.}(\alpha\mu\nu)=  \label{58} \\
& \partial_{\sigma}\eta_{\alpha}^{1}\Omega_{\rho\sigma}^{0}\partial_{\rho
}\left( \partial_{\beta}\eta_{\mu}^{1}\Omega_{\beta\nu}^{0}+\Omega_{\mu\beta
}^{0}\partial_{\beta}\eta_{\nu}^{1}\right) +\partial_{\sigma}\eta_{\alpha
}^{1}\left(
\Omega_{\nu\beta}^{0}\partial_{\beta}\Omega_{\sigma\mu}^{1}+\Omega_{\mu%
\beta}^{0}\partial_{\beta}\Omega_{\nu\sigma}^{1}\right) +  \notag \\
& \Omega_{\sigma\alpha}^{1}\partial_{\sigma}\left( \partial_{\beta}\eta
_{\mu}^{1}\Omega_{\beta\nu}^{0}+\Omega_{\mu\beta}^{0}\partial_{\beta}\eta
_{\nu}^{1}\right) +\mathrm{cycl.}(\alpha\mu\nu).  \notag
\end{align}
Then, taking into account (\ref{46}) and (\ref{56}), we represent (\ref{58})
as%
\begin{equation*}
\Omega_{\alpha\sigma}^{1}\partial_{\sigma}\Omega_{\mu\nu}^{1}+\mathrm{cycl.}%
(\alpha\mu\nu).
\end{equation*}
Finally, for (\ref{57}) one gets%
\begin{equation}
\Omega_{\alpha\beta}^{0}\partial_{\beta}\Omega_{\mu\nu}^{2}+\Omega
_{\alpha\beta}^{1}\partial_{\beta}\Omega_{\mu\nu}^{1}+\mathrm{cycl.}(\alpha
\mu\nu)=0,  \label{60}
\end{equation}
which is a Jacobi identity (\ref{6}) in the second order.

In the appendix \ref{AppA} we will prove by the induction that the
integrability condition (\ref{50}) in the $n$-th order is satisfied as a
consequence of a Jacobi identity (\ref{6}), like in the first two orders.

The solutions of the eqs. (\ref{51}) and. (\ref{51a}) are given by%
\begin{align}
y_{n}^{i} & =\int_{0}^{1}p_{j}F_{ji}^{n}\left( x,sp\right) sds+\partial
_{i}^{p}f^{n}\left( x,p\right) ,  \label{63} \\
\pi_{i}^{n} & =\int_{0}^{1}x^{j}F_{N+jN+i}^{n}\left( sx,p\right)
sds+\partial_{i}g^{n}\left( x,p\right) ,  \label{64}
\end{align}
correspondingly, where $f^{n}\left( x,p\right) $ and $g^{n}\left( x,p\right)
$ are arbitrary functions. Then, from (\ref{53}) we may define the relation
between $f^{n}$ and $g^{n}$ in terms of $F_{ij}^{n},\ F_{N+iN+j}^{n}$ and $%
F_{N+ij}^{n}$. So, the arbitrariness of the solution of the eq. (\ref{49})
is described by only one function%
\begin{equation}
g\left( x,p,\theta\right) =\sum_{n=0}^{\infty}\theta^{n}g^{n}\left(
x,p\right) .  \label{65}
\end{equation}
Since the addition of this function in (\ref{43}) does not change the
canonical PB (\ref{41}), it can be interpreted as a generating function of
the canonical transformation.

\section{Particular case, $J_{i}=p_{i}+\partial_{i}r\left( x,p,\protect\theta%
\right) .$}

Let us consider the particular case of our construction, choosing the vector
$J_{i}=p_{i}+\partial_{i}r\left( x,p,\theta\right) .$ In this case, by the
definition (\ref{21b}), $\bar{\omega}^{ij}=0$, and the matrix
\begin{equation}
\Omega=\left(
\begin{array}{cc}
\left( \bar{\delta}^{-1}\right) ^{T}\overline{\varpi}\bar{\delta }%
^{-1} & -\left( \bar{\delta}^{-1}\right) ^{T} \\
\bar{\delta}^{-1} & 0%
\end{array}
\right) .  \label{66}
\end{equation}
That is,%
\begin{equation}
\left\{ p_{i},p_{j}\right\} =\varpi^{ij}\left( x,p\right) =0.  \label{67}
\end{equation}
In section 3 we proved that such $\Omega$ exists and gave an iterative
procedure of its construction. In particular, from eqs. (\ref{27}) and (\ref%
{30}) one finds
\begin{equation}
J_{N+i}=\frac{\theta}{2}\omega^{il}p_{l}+\frac{\theta^{2}}{6}\omega
^{lk}\partial_{k}\omega^{mi}p_{l}p_{m}+O\left( \theta^{3}\right) ,
\label{68}
\end{equation}
where the function $f\left( x,p,\theta\right) $ which describe the
arbitrariness in the construction of $J_{N+i}$ was chosen to cancel the
contribution from $r\left( x,p,\theta\right) $, i.e., $f=r$. And then we
calculate%
\begin{align}
& \delta^{ij}\left( x,p\right) =\delta^{ij}+\frac{\theta}{2}\partial
_{j}\omega^{il}p_{l}+  \label{69} \\
& \theta^{2}\left( \frac{1}{12}\partial_{j}\omega^{kl}\partial_{k}%
\omega^{im}+\frac{1}{6}\omega^{lk}\partial_{j}\partial_{k}\omega^{im}\right)
p_{l}p_{m}+O\left( \theta^{3}\right) .  \notag
\end{align}
From the eqs. (\ref{36})-(\ref{38}) and (\ref{66}) we may see that the
structure of $J_{N+i}^{n},$ and $\delta_{n}^{ij}$ for $n>2$ will be the same
as in first two orders: they will be polynomials in $p$ of order $n$.

The next step is to construct the Darboux coordinates. From (\ref{49}) and (%
\ref{64}) we may see that choosing $g\left( x,p,\theta\right) =0$, one
obtains that $\pi_{i}^{n}=0,\ n\geq1$, i.e.,
\begin{equation}
\pi_{i}=p_{i}.  \label{70}
\end{equation}
For the coordinates $y^{i}$ one has%
\begin{equation}
y^{i}=x^{i}+\frac{\theta}{2}\omega^{ij}\left( x\right) p_{j}+\frac {%
\theta^{2}}{12}\left(
\omega^{lk}\partial_{l}\omega^{ij}+\omega^{lj}\partial_{l}\omega^{ik}\right)
p_{j}p_{k}+O\left( \theta^{3}\right) .  \label{71}
\end{equation}
Note that since $\delta_{n}^{ij}$ are polynomials in $p$ of order $n$, from (%
\ref{51}) and (\ref{53}) we conclude that each $y_{n}^{i}$ with $n\geq1$
also will be a polynomial in $p$ of order $n$. Making the inverse
transformation we find that the expression for $x^{i}$ in terms of $y^{i}$
and $\pi_{i}$ has a form%
\begin{equation}
x^{i}=y^{i}+\sum_{n=1}^{\infty}\Gamma^{i\left( n\right) }\left( y\right)
\left( \theta\pi\right) ^{n},  \label{72}
\end{equation}
where $\Gamma^{i\left( n\right) }\left( y\right)
=\Gamma^{ij_{1}...j_{n}}\left( y\right) .$ That is, each $x_{n}^{i}$ is a
polynomial in $\pi$ of order $n$. In particular, up to the second order one
has:%
\begin{equation}
x^{i}=y^{i}+\theta\Gamma^{ij}\left( y\right) \pi_{j}+\theta^{2}\Gamma
^{ijk}\left( y\right) \pi_{j}\pi_{k}+O\left( \theta^{3}\right) .  \label{73}
\end{equation}
The coefficients $\Gamma^{ij_{1}...j_{n}}\left( y\right) $ can be found
using the standard procedure described in two previous sections, using the
explicit form of $\delta_{n}^{ij}$ and formulas (\ref{63}) and (\ref{53}) to
construct $y_{n}^{i}$, and then constructing the inverse transformation.

However, in this case there is a much simpler procedure, based on the form (%
\ref{72}) of coordinates $x^{i}$, Poisson brackets (\ref{11}) with $\mu =i$
and $\nu =j$, i.e., $\left\{ x^{i},x^{j}\right\} =\theta \omega ^{ij}\left(
x\right) ,$ and the definition of the Darboux coordinates (\ref{41}). First
we write%
\begin{align}
& \left\{ y^{i}+\sum_{n=1}^{\infty }\Gamma ^{i\left( n\right) }\left(
y\right) \left( \theta \pi \right) ^{n},y^{j}+\sum_{n=1}^{\infty }\Gamma
^{j\left( n\right) }\left( y\right) \left( \theta \pi \right) ^{n}\right\} =
\label{74} \\
& \theta \omega ^{ij}\left( y^{i}+\sum_{n=1}^{\infty }\Gamma ^{i\left(
n\right) }\left( y\right) \left( \theta \pi \right) ^{n}\right) .  \notag
\end{align}%
Then, equating coefficients in the left and in the right-hand sides of (\ref%
{74}) in each order in $\theta ,$ one obtains algebraic equations on the
coefficients $\Gamma ^{i\left( n\right) }\left( y\right) $ in terms of $%
\omega ^{ij}$ and lower order coefficients $\Gamma ^{i\left( m\right)
}\left( y\right) ,\ m<n$. The existence of the solution of these equations
is a consequence of the combination of three facts: the existence of a
symplectic structure $\Omega _{\mu \nu }\left( \xi \right) =\Omega _{\mu \nu
}^{0}+O\left( \theta \right) $, such that $\Omega _{ij}\left( \xi \right)
=\theta \omega ^{ij}\left( x\right) $; the existence of a perturbative form (%
\ref{42}) of a Darboux coordinates and the fact that coordinates $x^{i}$
have a polynomial form (\ref{72}) in terms of $\pi _{j}$. The procedure is
analogous to the one proposed in \cite{KV} for the construction of
polydifferential representation of the algebra $[\hat{x}^{j},\hat{x}%
^{k}]=2\alpha \hat{\omega}^{jk}(\hat{x})$, where also was given a solution
of these algebraic equations.

In the first order in $\theta $ we have from (\ref{74}):%
\begin{equation*}
\Gamma ^{ji}-\Gamma ^{ij}=\omega ^{ij},
\end{equation*}%
with a solution $\Gamma ^{ij}=-\omega ^{ij}/2+s^{ij}$, where $s^{ij}$ is an
arbitrary symmetric matrix. Choosing $s^{ij}=0$, which correspond to the
choice $f=r$, we end up with $\Gamma ^{ij}=-\omega ^{ij}/2$. The equation in
the second order is:%
\begin{equation}
8\left( \Gamma ^{ijk}-\Gamma ^{jik}\right) =-\omega ^{mk}\partial _{m}\omega
^{ij}.  \label{75}
\end{equation}%
With a solution:%
\begin{equation}
\Gamma ^{ijk}=\frac{1}{24}\omega ^{km}\partial _{m}\omega ^{ij}+\frac{1}{24}%
\omega ^{jm}\partial _{m}\omega ^{ik}.  \label{76}
\end{equation}%
That is, up to the second order the expression for $x^{i}$ reads
\begin{equation}
x^{i}=y^{i}-\frac{\theta }{2}\omega ^{ij}\pi _{j}+\frac{\theta ^{2}}{24}%
\left( \omega ^{km}\partial _{m}\omega ^{ij}+\omega ^{jm}\partial _{m}\omega
^{ik}\right) \pi _{j}\pi _{k}+O\left( \theta ^{3}\right) .  \label{77}
\end{equation}%
One can verify that the inverse transformation of (\ref{71}) gives exactly
the eq. (\ref{77}).

In the $n$-th order the algebraic equation for $\Gamma ^{i\left( n\right)
}\left( y\right) $ reads%
\begin{equation}
n\Gamma ^{\lbrack ij]i_{2}\dots i_{n}}\left( y\right) =G^{iji_{2}\dots
i_{n}}\left( y\right) ,\ n>1,  \label{e1}
\end{equation}%
where%
\begin{eqnarray}
&&G^{iji_{2}\dots i_{n}}\pi _{i_{2}}...\pi _{i_{n}}=\omega
_{n-1}^{ij}-\sum_{m=1}^{n-1}\left\{ \Gamma ^{i\left( n-m\right) }\left( \pi
\right) ^{n-m},\Gamma ^{j\left( m\right) }\left( \pi \right) ^{m}\right\} ,
\label{e2} \\
&&\omega _{n-1}^{ij}=\left. \frac{d^{n-1}}{d\theta ^{n-1}}\omega ^{ij}\left(
y^{i}+\sum_{n=1}^{\infty }\Gamma ^{i\left( n\right) }\left( y\right) \left(
\theta \pi \right) ^{n}\right) \right\vert _{\theta =0}.  \notag
\end{eqnarray}%
The solution of the eq. (\ref{e1}) is given by%
\begin{equation}
\Gamma ^{ji_{1}\dots i_{n}}=\frac{1}{n(n+1)}\left( G^{ji_{1}i_{2}\dots
i_{n}}+G^{ji_{2}i_{1}i_{3}\dots i_{n}}+\dots +G^{ji_{n}i_{1}i_{2}\dots
i_{n-1}}\right) .  \label{e3}
\end{equation}

Once we know the explicit expression for the Darboux coordinates: $\eta
_{\mu }=\eta _{\mu }\left( \xi \right) $, we can invert it: $\xi _{\mu }=\xi
_{\mu }\left( \eta \right) $ and find the explicit form for the symplectic
structure%
\begin{equation}
\Omega _{\mu \nu }\left( \xi \right) =\left\{ \xi _{\mu }\left( \eta \right)
,\xi _{\nu }\left( \eta \right) \right\} |_{\eta _{\mu }=\eta _{\mu }\left(
\xi \right) }.  \label{Omega1}
\end{equation}%
By the construction (\ref{Omega1}) obey the Jacobi identity (\ref{6}).

\section{General solution}

Now suppose that $J_{i}(x,p,\theta)$ is any arbitrary function. It means
that the functions $\varpi^{ij}\left( x,p\right) $ and $\delta^{ij}\left(
x,p\right) $ will have expressions different from (\ref{67}) and (\ref{69}):%
\begin{align}
& \delta^{ji}\left( x,p\right) =\delta^{ji}+\theta(\partial_{i}^{p}J_{j}^{1}-%
\frac{1}{2}\partial_{j}\omega^{li}p_{l})+O\left( \theta^{2}\right) ,
\label{79} \\
& \varpi^{ij}\left( x,p\right) =\theta\left(
\partial_{i}J_{j}^{1}-\partial_{j}J_{i}^{1}\right) +O\left(
\theta^{2}\right) .  \notag
\end{align}
Consequently, the expressions for $y^{i}$ and $\pi_{i}$ in terms of $x^{i}$
and $p_{i}$ will also change. However, the expression (\ref{72}) for $x^{i}$
in terms of $y^{i}$ and $\pi_{i}$ will remain the same, since $\Omega
_{ij}=\theta\omega^{ij}\left( x\right) $ does not change, and the (\ref{72})
continue obeying the eq.
\begin{equation*}
\left\{ x^{i}(y,\pi),x^{j}(y,\pi)\right\} =\theta\omega^{ij}\left(
x(y,\pi)\right) .
\end{equation*}
What will change now is the expression for $p_{i}$. In the general case we
write
\begin{equation}
p_{i}=\pi_{i}- j_{i}(y,\pi,\theta),\ \ \ \ j_{i}(y,\pi,\theta)=\partial_{i}f(y)+\theta j^1_{i}(y,\pi)+O\left( \theta^{2}\right).  \label{80}
\end{equation}
Which together with (\ref{77}) implies that%
\begin{equation}
\pi_{i}=p_{i}+\theta j_{i}^{1}(x,p)+\theta^{2}\left( \frac{1}{2}\partial
_{l}j_{i}^{1}\omega^{lk}p_{k}+\partial_{l}^{p}j_{i}^{1}j_{l}^{1}\right)
+O\left( \theta^{3}\right) .  \label{81}
\end{equation}
The Poisson brackets between $x^{i}(y,\pi)$ and $p_{i}(y,\pi)$ are%
\begin{align}
& \left\{ x^{i}(y,\pi),p_{j}(y,\pi)\right\} =\delta^{ij}+\theta
(\partial_{i}^{\pi}j_{j}^{1}(y,\pi)-\frac{1}{2}\partial_{j}\omega^{il}(y)%
\pi_{l})-  \label{82} \\
& \frac{\theta^{2}}{2}\left( \partial_{l}\omega^{ik}(y)\pi_{k}\partial
_{l}^{p}j_{j}^{1}(y,\pi)-\omega^{il}(y)\partial_{l}j_{j}^{1}(y,\pi)\right)
+O\left( \theta^{3}\right) ,  \notag
\end{align}%
\begin{align}
& \left\{ p_{i}(y,\pi),p_{j}(y,\pi)\right\}
=\theta(\partial_{j}j_{i}^{1}\left( y,\pi\right)
-\partial_{i}j_{j}^{1}\left( y,\pi\right) )+ \\
& \theta^{2}\left( \partial_{j}j_{i}^{2}\left( y,\pi\right) -\partial
_{i}j_{j}^{2}\left( y,\pi\right) +\partial_{l}j_{i}^{1}\left( y,\pi\right)
\partial_{l}^{p}j_{j}^{1}\left( y,\pi\right)
-\partial_{l}^{p}j_{i}^{1}\left( y,\pi\right) \partial_{l}j_{j}^{1}\left(
y,\pi\right) \right) +O\left( \theta^{3}\right) .  \notag
\end{align}
Taking into account (\ref{77}) and (\ref{81}) we find that%
\begin{align}
& \left\{ x^{i},p_{j}\right\} =\delta^{ij}\left( x,p\right) ,  \label{83} \\
& \left\{ p_{i},p_{j}\right\} =\varpi^{ij}\left( x,p\right) ,  \notag
\end{align}
where
\begin{align}
& \delta^{ij}\left( x,p\right) =\delta^{ij}+\theta(\partial_{i}^{p}j_{j}^{1}-%
\frac{1}{2}\partial_{j}\omega^{li}p_{l})  \label{84} \\
& +\frac{\theta^{2}}{2}\left( \partial_{j}\omega^{il}j_{l}-\frac{1}{2}%
\partial_{m}\partial_{j}\omega^{il}\omega^{mk}p_{k}p_{l}+\frac{1}{2}%
\partial_{j}\Gamma_{0}^{ikl}p_{k}p_{l}+\omega^{il}\partial_{l}j_{j}-%
\partial_{l}\omega^{ik}\partial_{l}^{p}j_{j}^{1}p_{k}\right) +O\left(
\theta^{3}\right) ,  \notag \\
& \varpi^{ij}\left( x,p\right) =\theta\left(
\partial_{j}j_{i}^{1}-\partial_{i}j_{j}^{1}\right) +\frac{\theta^{2}}{2}%
\left( \partial _{l}\left(
\partial_{j}j_{i}^{1}-\partial_{i}j_{j}^{1}\right) \omega
^{lk}p_{k}-2\partial_{l}^{p}\left(
\partial_{j}j_{i}^{1}-\partial_{i}j_{j}^{1}\right) j_{l}^{1}\right) +  \notag
\\
& \theta^{2}\left( \partial_{i}j_{j}^{2}-\partial_{j}j_{i}^{2}+\partial
_{l}j_{i}^{1}\partial_{l}^{p}j_{j}^{1}-\partial_{l}^{p}j_{i}^{1}\partial
_{l}j_{j}^{1}\right) +O\left( \theta^{3}\right) .  \notag
\end{align}
By the construction the Poisson brackets (\ref{83}) together with $\left\{
x^{i},x^{j}\right\} =\theta\omega^{ij}\left( x\right) $, obey the Jacobi
identity and equations (\ref{72}) and (\ref{80}) express phase space
variables $x^{i}$ and $p_{i}$ in terms of Darboux coordinates $y^{i}$ and $%
\pi_{i}.$ In fact, this process of construction of functions $%
\delta^{ij}\left( x,p\right) $ and $\varpi^{ij}\left( x,p\right) $ and
corresponding Darboux coordinates is much simpler then one described in
sections 3 and 4. However, the existence of this perturbative procedure is a
consequence of the existence of the procedures described in previous
sections.

As far as the description of the arbitrariness in solution is concerned, we
see comparing (\ref{79}) and (\ref{84}) that $j_{i}^{1}=J_{i}^{1}$ and $%
j_{i}^{2}$ can be expressed in terms of $J_{i}^{2}$ and combinations of $%
J_{i}^{1}$, $\omega^{lk}$ and its derivatives, etc. That is, the vector $%
J_{i}$ describing the arbitrariness in the construction of the symplectic
structure $\Omega_{\mu\nu}\left( \xi\right) $ can be expressed in terms of
vector $j_{i}$ and vice versa. So, in this direct method of the construction
of the Darboux coordinates and the symplectic structure the arbitrariness is
described by the vector $j_{i}(y,\pi,\theta)$.

\section{Canonical quantization and the star product}

In this section we will discuss the quantization of the obtained in symplectic manifold. From a mathematical point of view the quantization of a symplectic realization of a Poisson manifold is a formal deformation
of a local symplectic groupoid. This problem was considered in \cite{CDF},
where a universal generating function for a formal symplectic groupoid was
provided in terms of Kontsevich trees and the Kontsevich weights \cite%
{Kontsevich}. However, it should be noted that there is no systematic way to compute the
Kontsevich weights beyond the second order in the deformation, see e.g. \cite%
{DS}.

We will construct the quantization in the Darboux coordinates,
using the expressions (\ref{77}) and (\ref{80}) of original momenta $p_{i}$
and coordinates $x^{i}$, choosing the normal ordering, operators of momenta $%
\hat{\pi}_{i}$ stand on the right of the operators of coordinates $\hat {y}%
^{i}$.

Since variables $\eta_{\mu}$ have the canonical PB (\ref{41}), the corresponding operators $\hat{\eta}_{\mu}=\left( \hat{y}%
^{i},\hat{\pi}_{i}\right) $ obey the standard Heisenberg algebra:
\begin{equation}
\left[ \hat{y}^{i},\hat{y}^{j}\right] =\left[ \hat{\pi}_{i},\hat{\pi}_{j}%
\right] =0~,\ \ \ \left[ \hat{y}^{i},\hat{\pi}_{j}\right] =i\delta _{j}^{i}.
\label{85}
\end{equation}
We choose the coordinate representation for this algebra, the operators of
momenta are the derivatives, $\hat{\pi}_{i}=-i\partial_{i}$, and the
operators of coordinates are the operators of multiplication, $\hat{y}%
^{i}=x^{i}$. Due to the normal ordering, we get for the operators $\hat{x}%
^{i}$ and $\hat {p}_{i}$ the following expression:%
\begin{align}
&\hat{p}_{i}^{Db}=-i\partial_{i}+ j_{i}\left( x,-i\partial,\theta\right)~,  \label{86} \\
&\hat{x}_{Db}^{i}  =x^{i}+\frac{i\theta}{2}\omega^{il}\partial_{l}-\theta
^{2}\Gamma^{ilm}\partial_{l}\partial_{m}-i\theta^{3}\Gamma^{ilmn}\partial
_{l}\partial_{m}\partial_{n}+O\left( \theta^{4}\right) ,  \notag
\end{align}
where the differential operator $ j_{i}\left( x,-i\partial,\theta\right) $ correspond to the vector $j_{i}(y,\pi,\theta)$
describing an arbitrariness in our construction. As we will see this arbitrariness may be fixed if to require that the operators $\hat{p}_{i}$ should
be self-adjoint.

Let us calculate the commutator%
\begin{align}
\left[ \hat{x}_{Db}^{i},\hat{x}_{Db}^{j}\right] & =i\theta\omega ^{ij}\left(
x\right) -\frac{\theta^{2}}{2}\omega^{kl}\partial_{k}\omega
^{ij}\partial_{l}-\frac{i\theta^{3}}{8}\omega^{nk}\omega^{ml}\partial
_{n}\partial_{m}\omega^{ij}\partial_{k}\partial_{l}.  \label{87}
\end{align}
Using the eqs. (\ref{86}) and supposing the normal ordering between coordinate and momentum operators, i.e., all $\hat{x}_{Db}^{i}$ should stay on the left from $\hat{p}_{i}^{Db}$, one can see that right hand side of the first commutator can be represented
as%
\begin{equation}
\left[ \hat{x}_{Db}^{i},\hat{x}_{Db}^{j}\right] =i\theta\hat{\omega}%
_{Db}^{ij}=i\theta\hat{\omega}^{ij}\left( \hat{x}_{Db}\right) +i\theta \hat{%
\omega}_{co}^{ij}\left( \hat{x}_{Db},\hat{p}^{Db}\right) ~,  \label{88}
\end{equation}
where%
\begin{equation}
\hat{\omega}_{co}^{ij}\left( \hat{x}_{Db},\hat{p}^{Db}\right) =\frac {%
i\theta^{2}}{3}\omega^{nk}\partial_{k}\omega^{ml}\partial_{n}\partial
_{m}\omega^{ij}\left( \hat{x}_{Db}\right) \hat{p}_{l}^{Db}+O\left(
\theta^{3}\right) ,  \label{89}
\end{equation}
stands for quantum corrections to $\hat{\omega}^{ij}\left( \hat{x}%
_{Db}\right) $, that are needed to make the algebra (\ref{3}) consistent
i.e., obeying the sufficient condition. Note that due to dependence of $\hat{%
\omega}_{Db}^{ij}$ on operators of momenta $\hat{p}^{Db}$, the operators of
coordinates $\hat{x}_{Db}^{i}$ do not form the subalgebra anymore, like it
was stated in the beginning (\ref{1}).

This problem can be solved, introducing the quantum corrections to the
coordinate operators%
\begin{equation}
\hat{x}^{i}=\hat{x}_{Db}^{i}+\hat{x}_{co}^{i},  \label{90}
\end{equation}%
where operator $\hat{x}_{co}^{i}$ should be chosen in the way to cancel the
dependence of the commutator between coordinates (\ref{90}) on momenta, i.e.,%
\begin{equation}
\left[ \hat{x}_{Db}^{i}+\hat{x}_{co}^{i},\hat{x}_{Db}^{j}+\hat{x}_{co}^{j}%
\right] =i\theta \hat{\omega}_{q}^{ij}\left( \hat{x}_{Db}+\hat{x}%
_{co}\right) ~.  \label{91}
\end{equation}%
Since corrections (\ref{89}) starts from the second order in $\theta $ and
are of first order in derivatives $\hat{p}_{l}^{Db}=-i\partial _{l}$ in this
order, the corrections $\hat{x}_{co}^{i}$ to the operators of coordinates $%
\hat{x}_{Db}^{i}$ should start from the third order in $\theta $ and be
quadratic in derivatives to cancel contribution from $\hat{\omega}%
_{co}^{ij}\left( \hat{x}_{Db},\hat{p}^{Db}\right) $, i.e.,
\begin{equation*}
\hat{x}_{co}^{i}=-i\theta ^{3}\Gamma _{1}^{ilm}\left( x\right) \partial
_{l}\partial _{m}+O\left( \theta ^{4}\right) .
\end{equation*}%
Therefore we write for the coordinate operators%
\begin{equation}
\hat{x}^{i}=x^{i}+\frac{i\theta }{2}\omega ^{il}\partial _{l}-\theta
^{2}\Gamma ^{ilm}\partial _{l}\partial _{m}-i\theta ^{3}\left( \Gamma
^{ilmn}\partial _{l}\partial _{m}\partial _{n}+\Gamma _{1}^{ilm}\partial
_{l}\partial _{m}\right) +O\left( \theta ^{4}\right) .  \label{92}
\end{equation}%
The commutator between these operators reads%
\begin{equation}
\left[ \hat{x}^{i},\hat{x}^{j}\right] =i\theta \hat{\omega}^{ij}\left( \hat{x%
}\right) +\frac{i\theta ^{3}}{3}\omega ^{nk}\partial _{k}\omega
^{ml}\partial _{n}\partial _{m}\omega ^{ij}\partial _{l}+i\theta ^{3}\Gamma
_{1}^{\left[ ij\right] l}\partial _{l}+O\left( \theta ^{4}\right) .
\label{93}
\end{equation}%
That is, if to choose $\Gamma _{1}^{ilm}$ obeying the equation%
\begin{equation}
\Gamma _{1}^{\left[ ij\right] l}=-\frac{1}{3}\omega ^{nk}\partial _{k}\omega
^{ml}\partial _{n}\partial _{m}\omega ^{ij},  \label{94}
\end{equation}%
with a solution
\begin{equation}
\Gamma _{1}^{ijk}=\frac{1}{6}\omega ^{nl}\partial _{l}\omega ^{mk}\partial
_{n}\partial _{m}\omega ^{ij}+\frac{1}{6}\omega ^{nl}\partial _{l}\omega
^{mj}\partial _{n}\partial _{m}\omega ^{ik},  \label{95}
\end{equation}%
the eq. (\ref{93}) takes the form%
\begin{equation}
\left[ \hat{x}^{i},\hat{x}^{j}\right] =i\theta \hat{\omega}^{ij}\left( \hat{x%
}\right) +O\left( \theta ^{4}\right) .  \label{96}
\end{equation}%
We see that new operators of coordinates $\hat{x}^{i}$ with coefficients $%
\Gamma _{1}^{ijk}$ defined in (\ref{95}) form the subalgebra (\ref{96}). It
should be noted that the construction of the corrections $\hat{x}_{co}^{i}$
in the next order in $\theta $ will also require the corrections $\hat{\omega%
}_{co}^{ij}\left( \hat{x}\right) $ to the operator $\hat{\omega}^{ij}\left(
\hat{x}\right) $. However, these corrections will depend only on $\hat{x}%
^{i}.$ In this case, the expressions for both $\hat{x}^{i}$ and $\hat{\omega}%
_{q}^{ij}\left( \hat{x}\right) =\hat{\omega}^{ij}\left( \hat{x}\right) +\hat{%
\omega}_{co}^{ij}\left( \hat{x}\right) $ coincide with corresponding
expressions from \cite{KV}. This fact admits us to define a star product on
the algebra (\ref{1}) by the standard rule,
\begin{equation}
f\left( \hat{x}\right) g\left( \hat{x}\right) =\left( f\star g\right) \left(
\hat{x}\right) ,  \label{97}
\end{equation}%
where
\begin{align}
& (f\star g)(x)=f\left( \hat{x}\right) g(x)=fg+\frac{i\theta }{2}\partial
_{i}f\omega ^{ij}\partial _{j}g  \label{98} \\
& -\frac{\theta ^{2}}{4}\left[ \frac{1}{2}\omega ^{ij}\omega ^{kl}\partial
_{i}\partial _{k}f\partial _{j}\partial _{l}g-\frac{1}{3}\omega
^{ij}\partial _{j}\omega ^{kl}\left( \partial _{i}\partial _{k}f\partial
_{l}g-\partial _{k}f\partial _{i}\partial _{l}g\right) \right]  \notag \\
& -\frac{i\theta ^{3}}{8}\left[ \frac{1}{3}\omega ^{nl}\partial _{l}\omega
^{mk}\partial _{n}\partial _{m}\omega ^{ij}\left( \partial _{i}f\partial
_{j}\partial _{k}g-\partial _{i}g\partial _{j}\partial _{k}f\right) \right.
\notag \\
& +\frac{1}{6}\omega ^{nk}\partial _{n}\omega ^{jm}\partial _{m}\omega
^{il}\left( \partial _{i}\partial _{j}f\partial _{k}\partial _{l}g-\partial
_{i}\partial _{j}g\partial _{k}\partial _{l}f\right)  \notag \\
& +\frac{1}{3}\omega ^{ln}\partial _{l}\omega ^{jm}\omega ^{ik}\left(
\partial _{i}\partial _{j}f\partial _{k}\partial _{n}\partial _{m}g-\partial
_{i}\partial _{j}g\partial _{k}\partial _{n}\partial _{m}f\right)  \notag \\
& +\frac{1}{6}\omega ^{jl}\omega ^{im}\omega ^{kn}\partial _{i}\partial
_{j}\partial _{k}f\partial _{l}\partial _{n}\partial _{m}g  \notag \\
& \left. +\frac{1}{6}\omega ^{nk}\omega ^{ml}\partial _{n}\partial
_{m}\omega ^{ij}\left( \partial _{i}f\partial _{j}\partial _{k}\partial
_{l}g-\partial _{i}g\partial _{j}\partial _{k}\partial _{l}f\right) \right]
~+O\left( \theta ^{4}\right) .  \notag
\end{align}%
For more details of this construction see \cite{KV}.

\section{Trace functional}

To formulate the QM on noncommutative spaces we also need an expression for the trace functional on the
algebra with a star product, i.e., a functional $Tr\left( f\right) =\int \mathbf{\Omega }\left( x\right) f\left( x\right) ,$ where $\mathbf{\Omega }\left(
x\right) $,  is an integration measure, satisfying
\begin{equation}
Tr\left( f\star g\right) -Tr\left( fg\right) =0.  \label{trace}
\end{equation}
The condition (\ref{trace}) implies the cyclic property of trace.
The existence of a trace functional for the Kontsevich star-product related to a Poisson bi-vector $\omega ^{ij}$ was proven in \cite{FS}. The recursive procedure for the construction of a trace was proposed in \cite{kup15}, which consists in the following steps. First to find a
function $\mu \left( x\right) $ such that
\begin{equation}
\partial _{i}\left( \mu \omega ^{ij}\right) =0,  \label{4a}
\end{equation}%
and to define the measure as $\mathbf{\Omega }\left( x\right) =d^{N}x\mu \left( x\right) ,$ that is,
\begin{equation}
Tr\left( f\right) =\int d^{N}x\mu \left( x\right) f\left( x\right) .
\label{5a}
\end{equation}%
Note that if $\det\omega^{ij}\neq0$,
a natural integration measure is
$
\mathbf{\Omega }\left( x\right) =dx^{N}/\sqrt{\left\vert \det\omega^{ij}(x)\right\vert }$.
Then, using the gauge freedom in the definition of the star product \cite%
{Kontsevich} to construct a new star product
\begin{equation}
f\star ^{\prime }g=D^{-1}\left( Df\star Dg\right) ,  \label{10a}
\end{equation}%
choosing a gauge operator $D$ in such a way that the condition (\ref{trace}) holds true
for this new product.
Because one may verify that (\ref{trace}) is not satisfied for the star-product (\ref{98}), defined in the previous section. In particular, substituting (\ref{98}) in (\ref{trace}) for two arbitrary functions $f$ and $g$ vanishing on the
infinity, in the second
order in $\theta $ in the left hand side of (\ref{trace}) one gets%
\begin{equation}
-\frac{\theta ^{2}}{4}\int d^{N}x\mu \left( x\right) \left[ \frac{1}{2}%
\omega ^{ij}\omega ^{kl}\partial _{i}\partial _{k}f\partial _{j}\partial
_{l}g-\frac{1}{3}\omega ^{ij}\partial _{j}\omega ^{kl}\left( \partial
_{i}\partial _{k}f\partial _{l}g-\partial _{k}f\partial _{i}\partial
_{l}g\right) \right]   \label{7a}
\end{equation}%
Integrating this expression by parts on $f$ and $g$ and using (\ref{4a}) we
rewrite it as%
\begin{equation}
\frac{\theta ^{2}}{12}\int d^{N}x\partial _{i}f\partial _{l}\left( \mu
\omega ^{ij}\partial _{j}\omega ^{lk}\right) \partial _{k}g,  \label{8a}
\end{equation}%
where the matrix $\partial _{l}\left( \mu \omega ^{ij}\partial _{j}\omega
^{lk}\right) $ is symmetric, i.e.,
\begin{equation}
\partial _{l}\left( \mu \omega ^{ij}\partial _{j}\omega ^{lk}\right)
=\partial _{l}\left( \mu \omega ^{kj}\partial _{j}\omega ^{li}\right) ,
\label{9a}
\end{equation}%
due to the JI and (\ref{4a}). The expression (\ref{8a}) is
different from zero.

Due to (\ref{8a}) and (\ref{9a}) we search the gauge operator $D$ in the form%
\begin{equation}
D=1+\theta ^{2}b^{ik}\partial _{i}\partial _{k}+O\left( \theta ^{3}\right) .
\label{11a}
\end{equation}%
In this case the new star product reads%
\begin{equation}
f\star ^{\prime }g=f\star g-2\theta ^{2}b^{ik}\partial _{i}f\partial
_{k}g+O\left( \theta ^{3}\right) .  \label{12a}
\end{equation}%
The condition (\ref{trace}) for the star product (\ref{12a}) in the second
order implies that%
\begin{equation}
\frac{\theta ^{2}}{12}\partial _{i}f\partial _{l}\left( \mu \omega
^{ij}\partial _{j}\omega ^{lk}\right) \partial _{k}g-2\theta ^{2}\mu
b^{ik}\partial _{i}f\partial _{k}g=0.  \label{13a}
\end{equation}%
That is,
\begin{equation}
b^{ik}=\frac{1}{24\mu }\partial _{l}\left( \mu \omega ^{ij}\partial
_{j}\omega ^{lk}\right) .  \label{14a}
\end{equation}%
We conclude that given a Poisson bi-vector $\omega ^{ij}$ and a function $%
\mu \left( x\right) $ obeying (\ref{4a}), the modified star product (\ref{12a}) admits
the trace (\ref{5a}).

\section{Quantization scheme and examples}

Now we have all necessary tools to define the consistent noncommutative quantum mechanics.

\textit{The Hilbert space} is determined as a space of
complex-valued functions which are square-integrable with a measure $\mathbf{%
\Omega }\left( x\right) $.

\textit{The internal product} between two states $%
\varphi \left( x\right) $ and $\psi \left( x\right) $ from the Hilbert space
is defined as
\begin{equation}
\left\langle \varphi \right\vert \left. \psi \right\rangle =Tr\left( \varphi
^{\ast }\star^{\prime } \psi \right) .  \label{scalar}
\end{equation}%

\textit{The action of the coordinate operators $\hat{x}^{i}$} on functions $\psi (x)$ from
the Hilbert space is defined through the modified star product (\ref{12a}), for any function $%
V\left( x\right) $ one has
\begin{equation}
{V}\left( \hat{x}\right) \psi (x)=V(x)\star^{\prime } \psi (x).  \label{15a}
\end{equation}%
The definitions (\ref{scalar}) and (\ref{15a}) means that the coordinate
operators are self-adjoint with respect to the introduced scalar product:
\begin{equation}
\langle\hat{x}^{i}\varphi|\psi\rangle=Tr\left(\left(x^i\star^{\prime }
\varphi\right)^{\ast}\star^{\prime }\psi\right)=Tr\left(
\varphi^{\ast}\star^{\prime } \left(x^i\star^{\prime }\psi\right)\right)=\langle\varphi|\hat{x}^{i}\psi\rangle.
\end{equation}

\textit{The momentum operators} $\hat{p}_{i}$ are fixed from the condition
that they also should be self-adjoint with respect to (\ref{scalar}). One of the
possibilities is to choose it in the form
\begin{equation}
\hat{p}_{i}=-i\partial _{i}-\frac{i}{2}\partial _{i}\ln \mu \left( x\right) .
\label{p}
\end{equation}%
One can easily verify that in this case $\left\langle \hat{p}_{i}\varphi
\right\vert \left. \psi \right\rangle =\left\langle \varphi \right\vert
\left. \hat{p}_{i}\psi \right\rangle .$

The choice (\ref{p}) for the representation of the momentum operators imply that they commute: $[\hat{p}_{i},\hat{p}_{j}]=0$. The commutator between $\hat{x}^{i}$ and $\hat{p}_{j}$ is
\begin{equation}
\left[ \hat{x}^{i},\hat{p}_{j}\right] =i\delta^i_j-\frac{i\theta}{2}\left( \partial_j\omega^{il}\left(\hat{x}%
\right)\hat{p}_{l}+ i\partial_j \left(\omega^{il}\partial_l\ln \mu\right) \left( \hat x\right) \right)+O\left( \theta ^{2}\right).  \label{xp}
\end{equation}
That is, as it was stated in the beginning the complete algebra of commutation relations involving $\hat{x}^{i}$ and $\hat{p}_{j}$ is a deformation in $\theta$ of a standard Heisenberg algebra.

\subsection{Free particle}

As an example we consider the eigenvalue problem for the Hamiltonian
\begin{equation}
\hat{H}=\frac{1}{2}\hat{p}_{i}\hat{p}^{i},  \label{103}
\end{equation}
describing a free particle. In this
case (\ref{103}) takes the form
\begin{equation}
\hat{H}=-\frac{1}{2}\left( \partial_{i}+\frac{1}{2}\partial _{i}\ln \mu \left( x\right) \right) ^{2},  \label{104}
\end{equation}
The eigenstates of this Hamiltonian are
\begin{equation}
\psi=e^{-ik_{i}x^{i}}\mu(x)^{-\frac{1}{2}},  \label{105}
\end{equation}
with eigenvalues%
\begin{equation}
E=\frac{1}{2}k_{i}k^{i},  \label{106}
\end{equation}
where $k_{i}$ are the eigenvalue of momenta $\hat{p}_{i}$. We see that the
spectrum of energy $E$ is non-negative and continuous like in undeformed
case, however the eigenstates differ from the plane wave on the commutative
space by the factor $\mu(x)^{-\frac{1}{2}}$, which may lead to different phenomenological consequences studying the processes of scattering of plane waves on curved noncommutative spaces.

\subsection{Three-dimensional isotropic harmonic oscillator}

It should be noted that the quantum mechanical scale of energies is rather
different from the Planck scale, therefore from the physical point of view it is useless to look for the effects caused by the noncommutativity in QM. However some important properties like
preservation of symmetries and corresponding consequences can be studied
already in QM. In particular,  it is well known fact that the canonical
noncommutativity, $\left[ \hat{x}^{i},\hat{x}^{j}\right] =i\theta ^{ij},$ breaks the
rotational symmetry of a particle in a central potential, which removes the degeneracy of the energy
levels \cite{Chaichian} over the magnetic quantum number $m$. This fact leads to the bounds of noncommutativity. The same logic remains in the field theory \cite{AGSV}. We will show here that the noncommutativity can be introduced in a way to preserve the symmetries and the corresponding degeneracy.

Let us consider three-dimensional isotropic harmonic oscillator described by the Hamiltonian
\begin{equation}
\hat{H}=\frac{\hat{p}^{2}}{2}+ \frac{\omega^2}{2}\hat{r}^{2}.
\label{H}
\end{equation}
where $r^{2}=x^{2}+y^{2}+z^{2}$. And let us choose the external antisymmetric field in a way to preserve the rotational symmetry of the system, $\omega ^{ij}(x)=\varepsilon ^{ijk}x^{k}$. The corresponding algebra of noncommutative coordinates is the algebra of fuzzy sphere \cite{Fuzzy},
\begin{equation}
\left[ \hat{x}^{i},\hat{x}^{j}\right] =i\theta\varepsilon ^{ijk}\hat{x}^{k}.\label{fuzzy}
\end{equation}
Note that the rotationally invariant NC space can be obtained as a foliation of fuzzy spheres \cite{Hammou}.

One may see that any function $\mu (r^{2})$ obeys the equation (\ref{4a}):
\begin{equation*}
\partial _{i}\left( \mu (r^{2})\varepsilon ^{ijk}x^{k}f\left( r^{2}\right)
\right) =0,
\end{equation*}%
and can be chosen as a measure to define a trace functional. For simplicity
we set $\mu \left( x\right) =1$. So, the momentum operators are just
derivatives $\hat{p}_{i}=-i\partial _{i}$. The specific choice of the external field implies that the coordinate operators transform as a vectors under rotations:
\begin{equation}
\left[ L_{i},\hat x^j \right] =i\varepsilon ^{ijk}x^{k},
\end{equation}%
where $L^{i}=-i\varepsilon ^{ijk}x_{j}\partial _{k}$ is the angular momentum operator.
That is, the rotational symmetry of (\ref{H}) will be conserved, preserving the degeneracy of the energy
spectrum over the magnetic quantum number $m$. The Hamiltonian can be written as
\begin{equation}
\hat{H}=-\frac{1}{2}\Delta+\frac{\omega^2}{2}{r}^{2}\star ^{\prime }= -\frac{1}{2}\Delta+\frac{\omega^2}{2}{r}^{2}+
\frac{\theta^{2} \omega^2}{24}L^{2}+O\left( \theta ^{4}\right) ,
\end{equation}%
where $L^{2}=L_{x}^{2}+L_{y}^{2}+L_{z}^{2}$ is the the orbital
momentum. We use the usual perturbation theory to calculate the leading corrections to
the energy levels,%
\begin{equation}
\Delta E_{n}^{NC}=\left\langle \psi ^{0}\left\vert \frac{\theta^{2} \omega^2}{24}L^{2}\right\vert
\psi ^{0}\right\rangle =\frac{\theta ^{2}\omega^{2}l(l+1)}{24},
\end{equation}%
where $\left\vert \psi ^{0}\right\rangle =R_{nl}(r)Y^m_l(\vartheta,\varphi) $ is
the unperturbed wave function, corresponding to the energy $%
E_{n}=\omega\left(n+\frac{3}{2}\right)$, here $n$ is a principal quantum number and $l$ is the azimuthal quantum number. The corresponding nonlocality is given by
\begin{equation}
\Delta x\Delta y\geq \frac{\theta ^{2}}{4}\left|m\right| ,
\end{equation}
where $m=-l,...,l$ and $l=0,1,...,n$. That is, the more the energy of the system the more the nonlocality.

\section{Conclusions and perspectives}

In conclusion we would like to overview some perspectives of the present
activity. As it was stated in the introduction our aim is to construct the
consistent quantum field theory on noncommutative spaces of general form. In
this connection the next step is to construct a relativistic generalization
of the proposed nonrelativistic quantum mechanics. If the external field $\omega^{\rho\sigma}(x)$ transforms as a two tensor with respect to a Lorentz group, e.g., $\omega^{\rho\sigma}(x)=\varepsilon^{\rho\sigma\lambda}x_\lambda f(x^2)$ in $(2+1)$ dimensions, the operators $\hat{x}^{\rho }=x^{\rho }+i\theta /2\omega ^{\rho\sigma }\partial
_{\sigma }+O\left( \theta ^{2}\right) $ and $\hat{p}_{\rho }=-i\partial _{\rho
}-i\partial _{\rho }\mu \left( x\right) $ will transform as vectors. This fact may be used to construct relativistic wave equations on coordinate dependent NC space-time. In particular, the free noncommutative Klein-Gordon equation can be written as
\begin{equation}
\left[ \left(\partial_\rho-\frac{1}{2}\partial_\rho \ln{\mu}\right)^2-m^2\right]\Phi=0. \label{KG}
\end{equation}
This equation is covariant under Lorentz transformation. The action leading to the equation (\ref{KG}) is
\begin{equation}
S_{free}=\int d^{N}x\mu \left( x\right)\left[\frac{1}{2}{\left(\hat p_\rho \Phi\right)}^{\ast}\star ^{\prime }\left(\hat p^\rho \Phi\right)+\frac{m^2}{2}\Phi^{\ast}\star ^{\prime }\Phi\right].
\end{equation}
One may also add the interaction term to this action, e.g., \begin{equation}S_{int}=\frac{\lambda}{4!}\int d^{N}x\mu \left( x\right)\Phi^{\ast}\star ^{\prime }\Phi\star ^{\prime }\Phi\star ^{\prime }\Phi^{\ast},\end{equation} and to study the corresponding quantum theory. Different questions may be addressed here like unitarity and renormalizability. An important problem here is to describe the physical mechanism which define the external field $\omega^{\rho\sigma}(x)$.

\section*{Acknowledgements}

I am grateful to Marcelo Gomes and Pedro Gomes for fruitful discussions and to Dima Vassilevich for useful comments.

%%%%%%%
\appendix

\section{Integrability condition for the Darboux coordinates}

\label{AppA} Let us consider the eq. (\ref{49}). In section 4 it was shown
that integrability condition for this equation in the first two orders in $%
\theta$ is satisfied as a consequence of the Jacobi identity (\ref{6}) for
the symplectic structure $\Omega_{\mu\nu}$. Suppose that the solution of (%
\ref{49}) was found up to the $(n-1)$-th order. Here we will prove that the
integrability condition (\ref{50}) in the $n$-th order is exactly Jacobi
identity (\ref{6}) in the $n$-th order. Let us write
\begin{align}
& L=\sum_{m=1}^{n-1}\Omega_{\alpha\beta}^{0}\partial_{\beta}\left[
\partial_{\sigma}\eta_{\mu}^{n-m}\Omega_{\sigma\rho}^{0}\partial_{\rho}%
\eta_{\nu}^{m}+\partial_{\sigma}\eta_{\mu}^{n-m}\Omega_{\sigma\nu}^{m}+%
\Omega_{\mu\sigma}^{n-m}\partial_{\sigma}\eta_{\nu}^{m}\right]  \label{b1} \\
& +\sum_{m=1}^{n-2}\sum_{k=1}^{n-m-1}\Omega_{\alpha\beta}^{0}\partial_{\beta
}\left(
\Omega_{\sigma\rho}^{n-m-k}\partial_{\sigma}\eta_{\mu}^{k}\partial_{\rho}%
\eta_{\nu}^{m}\right) +\mathrm{cycl.}(\alpha\mu\nu).  \notag
\end{align}
Using JI (\ref{40}) in the first line of (\ref{b1}) one rewrites $L$ as%
\begin{align}
& L=  \label{b2} \\
& \sum_{m=1}^{n-1}\left[ \partial_{\sigma}\eta_{\mu}^{m}\Omega_{\rho\sigma
}^{0}\partial_{\rho}\left( \partial_{\beta}\eta_{\nu}^{n-m}\Omega
_{\beta\alpha}^{0}+\Omega_{\nu\beta}^{0}\partial_{\beta}\eta_{\alpha}^{n-m}+%
\Omega_{\nu\alpha}^{n-m}\right) -\Omega_{\mu\sigma}^{m}\partial
_{\sigma}\left(
\partial_{\beta}\eta_{\nu}^{n-m}\Omega_{\beta\alpha}^{0}+\Omega_{\nu%
\beta}^{0}\partial_{\beta}\eta_{\alpha}^{n-m}\right) \right]  \notag \\
& +\sum_{m=1}^{n-2}\sum_{k=1}^{n-m-1}\left[ \Omega_{\alpha\beta}^{0}%
\partial_{\beta}\Omega_{\sigma\rho}^{n-m-k}\partial_{\sigma}\eta_{\mu}^{k}%
\partial_{\rho}\eta_{\nu}^{m}+\partial_{\sigma}\eta_{\mu}^{k}\Omega
_{\rho\sigma}^{n-m-k}\partial_{\rho}\left(
\Omega_{\nu\beta}^{0}\partial_{\beta}\eta_{\alpha}^{k}+\Omega_{\beta%
\alpha}^{0}\partial_{\beta}\eta_{\nu}^{k}\right) \right.  \notag \\
& \left. -\partial_{\sigma}\eta_{\mu}^{m}\left( \Omega_{\alpha\beta
}^{n-m-k}\partial_{\beta}\Omega_{\sigma\nu}^{k}+\mathrm{cycl.}(\alpha\sigma
\nu)\right) \right] +\mathrm{cycl.}(\alpha\mu\nu).  \notag
\end{align}
Then, taking into account the eq. (\ref{49}) in the first line of (\ref{b2})
we represent it in the form%
\begin{align*}
& L=\sum_{m=1}^{n-1}\Omega_{\mu\sigma}^{m}\partial_{\sigma}\Omega_{\nu\alpha
}^{n-m}+\sum_{m=1}^{n-2}\left\{ \partial_{\sigma}\eta_{\mu}^{m}\Omega
_{\rho\sigma}^{0}\partial_{\rho}\left(
\partial_{\beta}\eta_{\nu}^{n-m}\Omega_{\beta\alpha}^{0}+\Omega_{\nu%
\beta}^{0}\partial_{\beta}\eta_{\alpha
}^{n-m}+\Omega_{\nu\alpha}^{n-m}\right) \right. \\
& +\sum_{k=1}^{n-m-1}\left[ \Omega_{\mu\sigma}^{m}\partial_{\sigma}\left(
\partial_{\beta}\eta_{\nu}^{n-m-k}\Omega_{\beta\gamma}^{0}\partial_{\gamma
}\eta_{\alpha}^{k}\right)
+\Omega_{\alpha\beta}^{0}\partial_{\beta}\Omega_{\sigma\rho}^{n-m-k}%
\partial_{\sigma}\eta_{\mu}^{k}\partial_{\rho}\eta_{\nu}^{m}\right. \\
& \left. \left.
+\partial_{\sigma}\eta_{\mu}^{m}\Omega_{\rho\sigma}^{n-m-k}\partial_{\rho}%
\left( \partial_{\beta}\eta_{\nu}^{k}\Omega
_{\beta\alpha}^{0}+\Omega_{\nu\beta}^{0}\partial_{\beta}\eta_{\alpha}^{k}+%
\Omega_{\nu\alpha}^{k}\right) \right] \right\} \\
& +\sum_{m=1}^{n-3}\sum_{k=1}^{n-m-2}\sum_{l=1}^{n-m-k-1}\Omega_{\mu\sigma
}^{m}\partial_{\sigma}\left( \Omega_{\beta\gamma}^{n-m-k-l}\partial_{\beta
}\eta_{\nu}^{l}\partial_{\gamma}\eta_{\alpha}^{k}\right) +\mathrm{cycl.}%
(\alpha\mu\nu).
\end{align*}
Using the JI in the second term of the second line and also the eq. (\ref{49}%
) for $n=1$, we have%
\begin{align*}
& L=\sum_{m=1}^{n-1}\Omega_{\mu\sigma}^{m}\partial_{\sigma}\Omega_{\nu\alpha
}^{n-m}+\sum_{m=1}^{n-2}\partial_{\sigma}\eta_{\mu}^{m}\Omega_{\rho%
\sigma}^{0}\partial_{\rho}\left[ \partial_{\beta}\eta_{\nu}^{n-m}\Omega_{%
\beta
\alpha}^{0}+\Omega_{\nu\beta}^{0}\partial_{\beta}\eta_{\alpha}^{n-m}+%
\Omega_{\nu\alpha}^{n-m}\right. \\
& +\sum_{k=1}^{n-m-1}\left. \left( \Omega_{\mu\beta}^{n-m-k}\partial
_{\beta}\eta_{\nu}^{k}+\Omega_{\beta\nu}^{n-m-k}\partial_{\beta}\eta_{%
\mu}^{k}\right) \right] \\
&
+\sum_{m=1}^{n-3}\sum_{k=1}^{n-m-2}\partial_{\sigma}\eta_{\mu}^{m}\Omega_{%
\rho\sigma}^{k}\partial_{\rho}\left\{ \left(
\Omega_{\nu\beta}^{0}\partial_{\beta}\eta_{\alpha}^{n-m-k}+\Omega_{\beta%
\alpha}^{0}\partial_{\beta}\eta_{\nu}^{n-m-k}+\Omega_{\nu\alpha}^{n-m-k}%
\right) \right. \\
& +\sum_{l=1}^{n-m-k-1}\left. \left(
\partial_{\beta}\eta_{\nu}^{l}\Omega_{\beta\alpha}^{n-m-k-l}+\Omega_{\nu%
\beta}^{n-m-k-l}\partial_{\beta}\eta_{\alpha}^{l}\right) \right\} +\mathrm{%
cycl.}(\alpha\mu\nu).
\end{align*}

We will need the following Lemma:

\begin{lemma}
\label{L4} The identities
\begin{align}
& M=  \label{i1} \\
& \sum_{m=1}^{n-3}\sum_{k=1}^{n-m-2}\sum_{l=1}^{n-m-k-1}\left[
\partial_{\sigma}\eta_{\mu}^{m}\Omega_{\rho\sigma}^{0}\partial_{\rho}\left(
\partial_{\beta}\eta_{\nu}^{l}\Omega_{\beta\gamma}^{n-m-k-l}\partial_{\gamma
}\eta_{\alpha}^{k}\right) \right.  \notag \\
& \left.
+\partial_{\sigma}\eta_{\mu}^{m}\Omega_{\rho\sigma}^{k}\partial_{\rho}\left(
\partial_{\beta}\eta_{\nu}^{n-m-k-l}\Omega_{\beta
\gamma}^{0}\partial_{\gamma}\eta_{\alpha}^{l}\right) \right]  \notag \\
&
+\sum_{m=1}^{n-4}\sum_{k=1}^{n-m-3}\sum_{l=1}^{n-m-k-2}%
\sum_{p=1}^{n-m-k-l-1}\partial_{\sigma}\eta_{\mu}^{m}\Omega_{\rho\sigma}^{p}%
\partial_{\rho}\left( \partial_{\beta}\eta_{\nu}^{l}\Omega_{\beta\gamma
}^{n-m-k-l-p}\partial_{\gamma}\eta_{\alpha}^{k}\right)  \notag \\
& +\mathrm{cycl.}(\alpha\mu\nu)=0,  \notag
\end{align}
and%
\begin{equation}
N=\sum_{m=1}^{n-2}\sum_{k=1}^{n-m-1}\partial_{\sigma}\eta_{\mu}^{m}\Omega_{%
\rho\sigma}^{0}\partial_{\rho}\left( \partial_{\beta}\eta_{\nu
}^{n-m-k}\Omega_{\beta\gamma}^{0}\partial_{\gamma}\eta_{\alpha}^{k}\right) +%
\mathrm{cycl.}(\alpha\mu\nu)=0.  \label{i2}
\end{equation}
hold true.
\end{lemma}

\begin{proof}
One can verify that%
\begin{align}
&
\sum_{m=1}^{n-3}\sum_{k=1}^{n-m-2}\sum_{l=1}^{n-m-k-1}\partial_{\sigma}%
\eta_{\mu}^{m}\Omega_{\rho\sigma}^{0}\partial_{\rho}\left( \partial_{\beta
}\eta_{\nu}^{l}\Omega_{\beta\gamma}^{n-m-k-l}\partial_{\gamma}\eta_{\alpha
}^{k}\right) +\mathrm{cycl.}(\alpha\mu\nu)=  \label{b3} \\
& \sum_{m=1}^{n-3}\sum_{k=1}^{n-m-2}\sum_{l=1}^{n-m-k-1}\left[
\partial_{\sigma}\eta_{\mu}^{m}\Omega_{\rho\sigma}^{0}\partial_{\beta}\eta_{%
\nu}^{l}\partial_{\rho}\Omega_{\beta\gamma}^{n-m-k-l}\partial_{\gamma
}\eta_{\alpha}^{k}+\partial_{\sigma}\eta_{\mu}^{m}\Omega_{\rho\sigma}^{k}%
\partial_{\rho}\left( \partial_{\beta}\eta_{\nu}^{n-m-k-l}\Omega
_{\beta\gamma}^{0}\partial_{\gamma}\eta_{\alpha}^{l}\right) \right]  \notag
\\
& +\mathrm{cycl.}(\alpha\mu\nu).  \notag
\end{align}
Also using (\ref{40}) one can see that%
\begin{align}
&
\sum_{m=1}^{n-3}\sum_{k=1}^{n-m-2}\sum_{l=1}^{n-m-k-1}\partial_{\sigma}%
\eta_{\mu}^{m}\Omega_{\rho\sigma}^{0}\partial_{\beta}\eta_{\nu}^{l}%
\partial_{\rho}\Omega_{\beta\gamma}^{n-m-k-l}\partial_{\gamma}\eta_{\alpha
}^{k}+\mathrm{cycl.}(\alpha\mu\nu)=  \label{b4} \\
&
\sum_{m=1}^{n-4}\sum_{k=1}^{n-m-3}\sum_{l=1}^{n-m-k-2}\sum_{p=1}^{n-m-k-l-1}%
\partial_{\sigma}\eta_{\mu}^{m}\Omega_{\rho\sigma}^{p}\partial_{\rho}\left(
\partial_{\beta}\eta_{\nu}^{l}\Omega_{\beta\gamma
}^{n-m-k-l-p}\partial_{\gamma}\eta_{\alpha}^{k}\right) +\mathrm{cycl.}%
(\alpha\mu\nu).  \notag
\end{align}
Substituting the right-hand side of the eq. (\ref{b4}) in the right hand
side of the eq. (\ref{b3}) one finds that (\ref{i1}) holds true. The proof
of (\ref{i2}) is straightforward.
\end{proof}

The sum $L+M+N$ using the eq. (\ref{49}) for $n=2$ can be written as
\begin{align*}
& L+M+N=\sum_{m=1}^{n-1}\Omega_{\mu\sigma}^{m}\partial_{\sigma}\Omega
_{\nu\alpha}^{n-m}+\sum_{m=1}^{n-3}\partial_{\sigma}\eta_{\mu}^{m}\Omega
_{\rho\sigma}^{0}\partial_{\rho}\left\{ \Omega_{\nu\beta}^{0}\partial_{\beta
}\eta_{\alpha}^{n-m}+\Omega_{\beta\alpha}^{0}\partial_{\beta}\eta_{%
\nu}^{n-m}+\Omega_{\nu\alpha}^{n-m}\right. \\
& +\sum_{k=1}^{n-m-1}\left(
\Omega_{\mu\beta}^{n-m-k}\partial_{\beta}\eta_{\nu}^{k}+\Omega_{\beta%
\nu}^{n-m-k}\partial_{\beta}\eta_{\mu}^{k}+\partial_{\beta}\eta_{%
\nu}^{n-m-k}\Omega_{\beta\gamma}^{0}\partial
_{\gamma}\eta_{\alpha}^{k}\right) \\
& +\sum_{k=1}^{n-m-2}\sum_{l=1}^{n-m-k-1}\left. \partial_{\beta}\eta_{\nu
}^{l}\Omega_{\beta\gamma}^{n-m-k-l}\partial_{\gamma}\eta_{\alpha}^{k}\right\}
\\
&
+\sum_{m=1}^{n-4}\sum_{k=1}^{n-m-3}\partial_{\sigma}\eta_{\mu}^{m}\Omega_{%
\rho\sigma}^{k}\partial_{\rho}\left\{
\Omega_{\nu\beta}^{0}\partial_{\beta}\eta_{\alpha}^{n-m-k}+\Omega_{\beta%
\alpha}^{0}\partial_{\beta
}\eta_{\nu}^{n-m-k}+\Omega_{\nu\alpha}^{n-m-k}\right. \\
& +\sum_{l=1}^{n-m-k-1}\left( \Omega_{\nu\beta}^{n-m-k-l}\partial_{\beta
}\eta_{\alpha}^{l}+\Omega_{\beta\alpha}^{n-m-k-l}\partial_{\beta}\eta_{\nu
}^{l}+\partial_{\beta}\eta_{\nu}^{n-m-k-l}\Omega_{\beta\gamma}^{0}\partial_{%
\gamma}\eta_{\alpha}^{l}\right) \\
& +\sum_{l=1}^{n-m-k-2}\sum_{p=1}^{n-m-k-l-1}\left.
\partial_{\beta}\eta_{\nu}^{l}\Omega_{\beta\gamma}^{n-m-k-l-p}\partial_{%
\gamma}\eta_{\alpha }^{k}\right\} +\mathrm{cycl.}(\alpha\mu\nu).
\end{align*}
Again using eq. (\ref{49}) we end up with%
\begin{equation*}
L+M+N=\sum_{m=1}^{n-1}\Omega_{\mu\sigma}^{m}\partial_{\sigma}\Omega_{\nu
\alpha}^{n-m}.
\end{equation*}
Since $M=N=0$ due to Lemma \ref{L4}, the integrability condition (\ref{50})
has the form
\begin{equation*}
\Omega_{\mu\sigma}^{0}\partial_{\sigma}\Omega_{\nu\alpha}^{n}+%
\sum_{m=1}^{n-1}\Omega_{\mu\sigma}^{m}\partial_{\sigma}\Omega_{\nu%
\alpha}^{n-m}+\mbox{cycl.}(\mu\nu\alpha)=0,
\end{equation*}
which is exactly Jacobi Identity (\ref{6}) in the $n$-th order.

\end{document}